\documentclass[a4paper,11pt]{article}
\usepackage{siunitx}
\usepackage{fullpage}
\usepackage{listings}
\usepackage{latexsym, amsfonts, amssymb, amsmath, amsthm, mathrsfs, mathtools, setspace, graphics, graphicx, bbm, float,bigints, calc}
\usepackage{enumerate}
\usepackage{caption}
\usepackage{indentfirst}
\usepackage{framed}
\usepackage{yhmath}
\usepackage{bm}
\usepackage{booktabs}

\usepackage{color}
\usepackage[dvipsnames]{xcolor}
\usepackage{pdfpages}

\usepackage{lineno}

\usepackage{geometry}
\RequirePackage{tikz}
\usetikzlibrary{calc,trees,positioning,arrows,chains,shapes.geometric,%
    decorations.pathreplacing,decorations.pathmorphing,shapes,%
    matrix,shapes.symbols}
\allowdisplaybreaks    

\usepackage[affil-it]{authblk}

\usepackage{hyperref}

\newtheorem{thm}{Theorem}[section]

\newtheorem{lem}[thm]{Lemma}

\newtheorem{defi}[thm]{Definition}

\newtheorem{prop}[thm]{Proposition}

\newcommand{\vertiii}[1]{{\left\vert\kern-0.25ex\left\vert\kern-0.25ex\left\vert #1 
    \right\vert\kern-0.25ex\right\vert\kern-0.25ex\right\vert}}

\newcommand{\PP}{\mathbb{P}}
\newcommand{\E}{\mathbb{E}}

\newcommand{\R}{\mathbb{R}}

\definecolor{darkgreen}{rgb}{0.0, 0.5, 0.0}

\begin{document}
\title{Rainfall is rough}
\author[1,2]{Thomas Deschatre}
\author[2,3,4]{Marc Hoffmann}
\author[3]{Mathieu Rosenbaum}

\affil[1]{EDF Lab}
\affil[2]{FiME Lab}
\affil[3]{Université Paris Dauphine--PSL}
\affil[4]{Institut Universitaire de France}
{   \makeatletter
    \renewcommand\AB@affilsepx{: \protect\Affilfont}
    \makeatother


    \makeatletter
    \renewcommand\AB@affilsepx{, \protect\Affilfont}
    \makeatother

    }

\maketitle
\maketitle


\begin{abstract}
We propose a new approach to model rainfall by combining heterogeneous data sources at different time scales. Continuous arrivals of rain cells are incorporated into a Hawkes process formalism that encompasses the classical Bartlett–Lewis and Neyman–Scott models, thereby providing a more flexible representation of clustering. Analysis of high-frequency rainfall data (at the minute scale over several years) indicates that critical Hawkes processes with heavy-tailed power-law kernels yield a superior fit relative to classical models and alternative kernel specifications. Scaling arguments inspired by \cite{JaissonRosenbaum2016} imply that aggregated rainfall at coarse time scales converges to a rough fractional process with Hurst exponent close to zero. This prediction is supported by empirical evidence from low-frequency data (annual observations spanning centuries to millennia), where the Hurst exponent is estimated to lie between $10^{-2}$ and $10^{-1}$ based on either direct observations from weather stations or proxy reconstructions such as tree-ring records. These results establish a connection between rainfall dynamics and models developed in quantitative finance for market microstructure and volatility. They also provide a new perspective on classical scaling phenomena originally studied by Hurst and Mandelbrot.
\end{abstract}

\vspace{2mm}

{\small 
}


\section{Introduction and main results}

Scaling laws and self-similarity properties are ubiquitous in many physical phenomena, and rainfall models in hydrology are no exception \cite{LeCam1961, Waymire1985Scaling}. Although scaling relations persist across scales, the statistical models that best describe the observations may differ substantially from one scale to another. This is particularly true for precipitation fields: at fine scales, precipitation is naturally described by clustered point-process models \cite{RodriguezIturbe1987}, whereas at coarser scales aggregated rainfall exhibits continuous self-similar behaviour \cite{SchertzerLovejoy1987, VerrierMalletBarthes2011}. Bridging these apparently different statistical descriptions remains a challenging problem. The present paper identifies a unified statistical framework that connects these two regimes by combining heterogeneous data sources while preserving key scaling characteristics.\\

The statistical inference methodology we introduce can identify critical cluster processes at small scales while recovering the signature of this criticality on coarser scale data obtained from other measurement sources. While we depart from classical approaches, our conclusions remain compatible with several well established findings. Interestingly, a similar conceptual shift has recently occurred in financial econometrics with the emergence of the theory of rough volatility \cite{Gatheral2022}. Our contribution can be viewed as an analogous paradigm for rainfall modelling.

\subsection{Modelling rainfall} \label{sec: mod rainfall gen}


Our primary interest lies in the dynamics of aggregated rainfall at small time scales. For a given weather station, we have observations at the scale of a few minutes over several years, see Section \ref{sec: stat criticality} for a description of the data. We model the data via
\begin{equation} \label{eq: data small scale}
Y_k^\delta = \int_{(k-1)\delta}^{k\delta}Y_s\,ds,\;\;k=1,2, \ldots
\end{equation}
where the latent process $Y_t = (Y_t)_{t \geq 0}$ is the intensity of rainfall and $\delta$ is on the order of a few minutes. In this widely accepted phenomenological paradigm, a stochastic counting
process $N_t = (N_t)_{t \geq 0}$ defines the arrival times of so-called rain cells. Each rain cell is associated with a rectangular pulse consisting of  a random intensity $I_i \geq 0$ times a random duration $L_i \geq 0$. The rainfall intensity $Y_t$ at time $t$ is the sum of all the intensities of living rain cells at time $t$:
\begin{equation} \label{eq:intensity_rain}  
\left\{
\begin{array}{ll}
Y_t  & = \sum_{i=1}^{N_t} I_i \,{\bf 1}_{\{\tau_i + L_i > t\}}, \\ \\
N_t & = \sum_{i \geq 1} {\bf 1}_{\{\tau_i \leq t\}},
\end{array}
\right.
\end{equation}
where the $(\tau_i)_{i \geq 1}$ are the arrival times of the rain cells.
\\ 

Continuous time modelling of rainfall dates back to \cite{LeCam1961}
and was popularized in the seminal papers by 
\cite{RodriguezIturbe1987, RodriguezIturbe1988}, see also \cite{OnofWheater1993, Onof2000, Cowpertwait2007}. Bartlett-Lewis (BL) and Neyman-Scott (NS) models are the most prominent models for rainfall intensity $Y_t$, see \cite{RodriguezIturbe1987}. They only differ by the specification of the counting process $N_t$, the pairs $(I_i, L_i)$ being drawn independently according to some pre-specified distribution. An advantage of this framework is the ability of $N_t$ to produce clusters, see \cite{DaleyVereJones2003} for a rigorous definition\footnote{Informally and consistently with the hydrology literature, we mean that the conditional intensity of observing an event increases following the occurrence of a previous event, resulting in temporal aggregation beyond that expected under a Poisson process.}.\\

 For the statistical analysis of rainfall at small scales, we have several time series of the form \eqref{eq: data small scale}, each time series being associated with a weather station at a given location over a specific observation period, see Section \ref{sec: stat criticality}. The fact that $Y_t$, hence $N_t$ is observed only via aggregated quantities of the form $Y_k^\delta = \int_{(k-1)\delta}^{k\delta} Y_s\,ds$ for $k=1, 2,\ldots$
creates a significant technical difficulty for statistical inference that we address in Section \ref{sec: inference small scales}  below.

\subsection{Hawkes processes, heavy tails and criticality} \label{sec: hawkes first time}

Among counting processes producing clusters, Hawkes processes \cite{Hawkes1971_Biometrika} are a natural candidate for extending the BL and NS models. They combine mathematical flexibility, transparent modelling interpretation and statistical tractability\footnote{This certainly explains the myriad of applications of Hawkes processes in a remarkably wide variety of fields ranging from finance (credit risk \cite{errais2010affine}, \cite{azizpour2018credit}; order book modelling \cite{abergel2015long, bacry2015hawkesob}; microstructure \cite{BacryDelattreHoffmannMuzy2013,el2018microstructural,el2019characteristic}; endogeneity \cite{filimonov2012quantifying,HardimanBercotBouchaud2013}; market impact and rough volatility \cite{JaissonRosenbaum2015, JaissonRosenbaum2016,  dandapani2021quadratic,jusselin2020no,chahdi2024theory}), to social networks \cite{zhao2015seismic}, life sciences, such as epidemiology \cite{rizoiu2017sirhawkes}, neuroscience \cite{reynaud2010adaptive, reynaud2014detecting, eichler2016graphical} and even criminology \cite{mohler2011self}.}. More importantly, Hawkes processes offer a new pathway to better understand the statistical nature of rainfall at fine scales via the notion of heavy-tailed criticality, see Section \ref{sec: intro from crit to rough} below. A counting process $N_t$ can be defined via its stochastic intensity $\lambda_t$:  
\begin{equation}\label{int_Hawkes}\PP(N_{t+dt}-N_t =1\,|\,N_s,s<t) = \lambda_t\,dt.
\end{equation}
The intensity of a linear Hawkes process is specified by
\begin{equation} \label{eq: def intensity Hawkes}
\lambda_t = \mu + \int_{[0,t)} \varphi(t-s)dN_s = \mu + \sum_{i, \tau_i <t} \varphi(t-\tau_i),
\end{equation}
where $\mu >0$ is a baseline Poisson intensity, the $\tau_i$ are the event times, and $\varphi: [0,\infty) \rightarrow [0,\infty)$ is a causal function that allows past events to trigger new events that can form clusters. 
Two prototypes of non-increasing reproduction kernels are given by the exponential kernel
$$\varphi(t) = \alpha_1 \exp(-\beta_1 t){\bf 1}_{\{t \geq 0\}},\;\;\alpha_1, \beta_1 >0,$$
and the power-law kernel
\begin{equation} \label{eq: power-law def}
\varphi(t) = \frac{\gamma}{(1+ \beta t)^{1+\alpha}}{\bf 1}_{\{t \geq 0\}},\;\; \alpha, \beta,\gamma >0.
\end{equation}
The two parametrizations yield quite different behaviours over large time intervals: typically, as $T \rightarrow \infty$,  $N_T$ scales like $T$  for an exponential kernel but it can scale like $T^{2\alpha}$ for a power-law kernel under additional conditions on $\mu$ and $\varphi$, see \cite{BacryDelattreHoffmannMuzy2013, JaissonRosenbaum2016} and see Section \ref{sec: connecting} for a mathematical description of these scaling laws.This will be of importance in our subsequent analysis.\\
 
Exponential kernels yield explicit moment formulas as well as autocorrelation functions via a spectral approach, see \cite{Hawkes1971_Biometrika, HawkesOakes1974}. We show in Proposition \ref{prop: expo second-order} in Section \ref{sec: model small scales} that second-order statistics computed from data \eqref{eq: data small scale} when the intensity process $(Y_t)_{t \geq 0}$ in \eqref{eq:intensity_rain} is governed by either the BL or the NS model are indistinguishable from those obtained under a model where $(N_t)_{t \geq 0}$ in \eqref{eq:intensity_rain} is taken as a Hawkes process with exponential kernel. Power-law kernels are significantly more flexible to enforce cluster at longer characteristic time scales. Intermediate between these situations are multiscale exponential kernels of the form
$$\varphi(t) = \sum_{k = 1}^d \alpha_k \exp(-\beta_k t){\bf 1}_{\{t \geq 0\}},\;\; \alpha_k,\beta_k >0,$$
see \cite{HardimanBercotBouchaud2013} in the context of Hawkes processes. For computational efficiency, power-law kernels are approximated by sums of exponentials to within prescribed accuracy, a common practice in numerical analysis and physics \cite{beylkin2005approximation}.\\

 We develop in Section \ref{sec: stat methodo micro} two statistical methods for inferring the parameters of exponential and power-law kernels (possibly approximated by sums of exponentials with suitable weights) from data \eqref{eq: data small scale}. The first approach is based on spectral methods  for Hawkes processes observed via count data as in \cite{CheyssonLang2022}, while the second method relies on the minimisation of second-order contrasts based on the aggregated quantities of $\E[Y_k^h],  \mathrm{Var}(Y_k^h)$, $\mathrm{Cov}(Y_k^h, Y_{k+1}^h)$, where $h$ is a multiple of $\delta$, {\it i.e.} a contrast across observable scales. 
In Section \ref{sec: stat criticality}, we analyse several datasets from weather stations, with a small-scale resolution $\delta$ ranging from 5 to 10 minutes and a total observation period spanning several years.\\


Our key finding is that power-law kernel systematically yield better statistical fits than other kernels or the BL and NS models, for different goodness-of-fit criteria. Moreover, we have strong evidence of criticality. This means that $\|\varphi\|_{1} =  \int_0^\infty \varphi(s)ds$ is close to $1$ (between 0.90 and 0.99). The parameter $\|\varphi\|_{1}$ is a natural measure of endogeneity\footnote{By endogeneity, we mean the ratio of events that are triggered via $\varphi$ against the total number of events.}, see Section \ref{sec: model small scales}, and our analysis empirically demonstrates criticality for rainfall at small scales, a phenomenon similar to the criticality observed in high-frequency transaction flows on financial markets \cite{HardimanBercotBouchaud2013}. 

\subsection{From heavy-tailed critical Hawkes to rough fractional  processes} \label{sec: intro from crit to rough}

Another important finding of Section \ref{sec: stat criticality} is that the estimators of the power-law tail parameter $\alpha$ are close to $0.5$ and away from $1$. Combined with criticality, the heavy-tailed behaviour of the kernel $\varphi$ is the gateway to roughness at large time scales: \cite{JaissonRosenbaum2016} prove that a nearly critical linear Hawkes process with a heavy-tailed kernel under a large-time scaling has a single non-degenerate limit behaviour as the integral of a continuous rough fractional process. This imposes $\mu \approx T^{\alpha-1}$ and $1-\|\varphi\|_{1} \approx T^{-\alpha}$ for some $\alpha \in (1/2, 1)$ in the representation of the power-law kernel \eqref{eq: power-law def}, see Definition \ref{def: Critical linear Hawkes processes} in Section \ref{sec: model small scales} below. This scaling yields\footnote{The precise meaning of the convergence is discussed in Section \ref{sec: connecting} below.}
\begin{equation} \label{eq: conv hawkes}
T^{-2\alpha}N_{tT} \longrightarrow \int_0^t X_s\,ds
\end{equation}
as $T\rightarrow \infty$,
where $X_t = (X_t)_{t \in [0,1]}$ is a rough fractional process with Hurst index 
$$H=\alpha -\tfrac{1}{2} \in (0,\tfrac{1}{2}).$$ By fractional, we mean a probabilistic self-similarity property of the form 
\begin{equation} \label{eq: scaling non linear approx}
\E\big[\big|X_{t+h}-X_t\big|^q\big] \approx h^{qH},\;\;q>0,\;h>0,
\end{equation}
and rough means $H \in (0,1/2)$ following the terminology of \cite{Gatheral2022}, see Definition \ref{def: fractional process} in Section \ref{sec: model large scale data} below. 
In particular, the H\"older smoothness saturates in a weak sense to the value $H$. Thanks to the convergence \eqref{eq: conv hawkes}, we establish in Section \ref{sec: connecting hawkes and fractional} a similar result for the rainfall intensity process, namely
 \begin{equation} \label{eq: aggregated limit}
T^{-2\alpha}\int_0^{tT} Y_s\,ds \longrightarrow \kappa \int_0^t X_s\,ds
\end{equation}
as $T\rightarrow \infty$, with $\kappa = \mathbb{E}[I]\mathbb{E}[L]$. In other words, the approximation \eqref{eq: aggregated limit} suggests that aggregated rainfall\footnote{{\it i.e.} increments of the left-hand side of \eqref{eq: aggregated limit}.} over large time resembles a rough fractional process over a fixed macroscopic time, see Proposition~\ref{prop: small meets large}. The argument is developed in Section \ref{sec: connecting}.

\subsection{Rainfall is rough at large time scales}

We undertake the verification of this theoretical prediction from a statistical angle in Section \ref{sec: rain is rough at large scales}.
For a fixed geographical area and a temporal aggregation scale $\Delta >0$ ranging from one year to one decade, we now have observations of the form 
\begin{equation} \label{eq: empirical data macro}
Z_1^\Delta, Z_2^\Delta, \ldots, Z_k^\Delta,\ldots
\end{equation}
Each value $Z_k^\Delta$ measures a proxy of the aggregation of rain during a period of size $\Delta$. The time period over which each time series is studied varies from several decades to a few thousand of years. For the subsequent mathematical analysis, it is convenient to define a continuous time embedding
$
Z_{k\Delta/T} = Z_k^\Delta,
$
where the continuous random process $(Z_t)_{t \in [0,1]}$ interpolates the time series $(Z_k^\Delta)_{k=1,\ldots}$ at the discrete times $k\Delta/ T$ whenever $T$ is the macroscopic time horizon over which the experiment is conducted. 
From empirical data \eqref{eq: empirical data macro}, we estimate\footnote{Whether a smoothness parameter can be inferred from noisy data is an old question in statistics, and property \eqref{eq: scaling non linear approx} gives a definite positive answer, since the scaling law provides some stability, see \cite{GloterHoffmann2007, ChongHoffmannLiuRosenbaumSzymanski2024, Szymanski2024} for a mathematical analysis of this phenomenon.} the expectation \eqref{eq: scaling non linear approx} for several values of $h$ as multiple of $\Delta/T$ via the empirical structure function
\begin{equation} \label{eq: empirical structure}
m_h(q, Z) = N^{-1} \sum_{k = 1}^{N}|Z_{kh}-Z_{(k-1)h}|^q,
\end{equation}
that are observable, up to boundary effects, thanks to the embedding $
Z_{k\Delta/T} = Z_k^\Delta$. For a fixed value of $q$, we obtain an estimator of $Hq$ by a linear regression on a log-log scale of the relation \eqref{eq: empirical structure} {\it i.e.} $\log m_h(q, Z)$ against $Hq\log h$ for different values of $h$. We then regress our estimators for different values of $q$ in a second time in order to obtain our final estimator of $H$, see Section \ref{sec: stat evidence roughness}.\\

We analyse several categories of data in Section~\ref{sec: evidence roughness}: i) Weather station data: observed data, including the four M\'et\'eo France station used in the empirical study at small time scales of Section~\ref{sec: stat criticality} and the data from~\cite{volpi2024} ii) Paleoclimatic data, from indirect measurements (tree rings, lake sediments, and pollen), used in the metastudy of \cite{iliopoulou2018}. Our key finding is that $H$ is empirically small, between $10^{-2}$ and $10^{-1}$, and that this result is robust.\\ 

These results are in accordance with the small scale study via the prediction of the scaling limit \eqref{eq: aggregated limit}, although the data we analyse at small and large scales are heterogeneous, measured with completely different methods. This thus reveals the emerging trace of heavy-tailed criticality in large time scales.
  
\subsection{Organisation of the paper}

In Section \ref{sec: rainfall at small scales}, we construct our model of rainfall intensity based on rain cells arrivals modelled as a linear Hawkes process. 
We show in particular in Proposition \ref{prop: expo second-order} that the Hawkes process formalism is already contained in the classical Neyman-Scott and Bartlett-Lewis models, and that these classical models are indistinguishable from a Hawkes model with an exponential kernel from first and second-order statistics.
In Section \ref{sec: stat methodo micro}, we develop two methods of inference for the parameters of the intensity of the model, based on a spectral approach for the first one and on second-order contrast minimization over several temporal scales for the second one.\\ 

In Section \ref{sec: connecting}, we connect our findings on Hawkes processes at small time scales by looking at large time scale data. 
By exploiting the results of \cite{JaissonRosenbaum2016}, we establish a formal link between heavy-tailed critical rainfall models at small time scales and their rough limiting behaviour, see in particular Proposition \ref{prop: small meets large}.\\

In Section \ref{sec: rain is rough at large scales}, we undertake the verification of the compatibility of Hawkes processes at small time scales with rough fractional processes at large time scales.  We briefly explain in Section \ref{sec: stat evidence roughness} the statistical methodology needed to recover the parameter $H$. We analyse large time scale data in Section \ref{sec: evidence roughness} for rainfall aggregation in several geographical areas over thousands of years, and we observe a systematic rough behaviour at large time scales.\\ 

In the discussion Section \ref{sec: discussion}, we compare our results with other more classical approaches, including the seminal results of \cite{hurst1951} and \cite{mandelbrot1968}. Also, the multifractal analysis that seems ubiquitous in physical phenomenona involving scaling laws is present in rainfall, see \cite{LovejoySchertzer2007}. We discuss how the multifractal approach compares with our framework. Interestingly, apparent statistical paradoxes about the presence of long memory or the multifractal nature of the data when confronted with the rough models in the context of rainfall have parallels in statistical finance, see \cite{Bacry2003LogInf, Gatheral2022,Neuman2018fBmZero, Wu2022LogSfBM}. We discuss how to reconcile seemingly different observations.\\ 

An appendix, Section \ref{sec: appendix}, contains the technical proofs and additional statistical results.

\section{Rainfall is critical and heavy-tailed at small time scales} \label{sec: rainfall at small scales}

\subsection{Modelling rainfall at small time scales} \label{sec: model small scales}

We generalize the classical approach of NS and BL models by allowing each rain cell to itself generate a new generation of rain cells, via linear Hawkes processes, as introduced in \cite{Hawkes1971_Biometrika}. 
The dynamics of the rain-cell arrival process $N_t$ can be described as follows:
\begin{itemize}
    \item[(i)] Parent rain cells arrive as a homogeneous Poisson process with intensity $\mu$;
    \item[(ii)] Each parent rain cell produces rain cells as an inhomogeneous Poisson process with intensity $\varphi(a)$, where $a$ is the age of the parent rain cell (the difference between current time $t$ and the birth of the parent rain cell);
       \item[(iii)] Each offspring rain cell gives rise to rain cell children as in step (ii), and so on.
\end{itemize}

The stochastic intensity $\lambda_t$ of a counting process $N_t$ is defined by Equation \eqref{int_Hawkes}.
In mathematically rigorous terms, this is equivalent to 
saying that 
$N_t-\int_0^t \lambda_s\,ds$ is a (local) martingale while $\int_0^t\lambda_sds$ is predictable. The intensity process $(\lambda_t)_{t \geq 0}$ entirely characterizes the law of $N_t = (N_t)_{t \geq 0}$, see {\it e.g.} \cite{Jacod1975} or \cite{Karr1991}. A linear Hawkes process $(N_t)_{t \geq 0}$ is defined via its stochastic intensity specified in Equation \eqref{eq: def intensity Hawkes}. 
A solution to \eqref{eq: def intensity Hawkes}  exists and is uniquely  defined as soon as $\varphi$ is locally integrable.\\

A key parameter is $\|\varphi\|_1 = \int_0^{\infty} \varphi(s)ds$, which accounts for the mean number of children produced by each rain cell in the population interpretation of $N_t$.
Throughout the rest of the paper, we assume the following natural stability condition, see {\it e.g. }\cite{JaissonRosenbaum2015}.
\begin{defi} \label{def: stability}
A linear Hawkes model is stable if $\|\varphi\|_{1} < 1$.
\end{defi}


\subsubsection*{Bartlett-Lewis, Neyman-Scott and the Hawkes formalism}
The BL and NS models \cite{RodriguezIturbe1987} are constructed as follows. For the BL model\footnote{For the BL and Hawkes models, the parent rain cells are included in $N_t$, but not for the NS process, as is customary in the original papers.} 
\begin{itemize}
    \item[(i)] Parent rain cells\footnote{Called storms in \cite{RodriguezIturbe1987}} appear as a homogeneous Poisson process with intensity $\mu_{BL} > 0$;
    \item[(ii)] During the lifetime of each parent rain cell, taken as exponentially distributed with parameter $\gamma_{BL} >0$, new rain cells appear as a second homogeneous Poisson process with intensity $\nu_{BL} > 0$.
\end{itemize}
The NS model has a similar structure: 
\begin{itemize}
    \item[(i')] Parent rain cells appear as a homogeneous Poisson process with intensity $\mu_{NS} > 0$;
    \item[(ii')] For each parent rain cell, a random number $C$ new rain cells are generated and independently displaced from the parent rain cell's origin, with a common distribution $f_{NS}(t)dt$.
\end{itemize}

The BL and NS models differ from the Hawkes process in that they involve a finite number of clustering levels. In contrast, in a Hawkes process, each generation of cells can trigger subsequent generations through step (iii). However, the arrival process of rain cells for both the BL and NS models actually follows a kind of degenerate Hawkes dynamics, with latent components. This observation for the NS model already appears in \cite{HawkesOakes1974}\footnote{Hawkes and Oakes in the proof of \cite[Lemma 1]{HawkesOakes1974} explain that for any arrival time of the Hawkes model, the cluster generated by that point is ``{\it a Neyman-Scott cluster process in which the number of events in a cluster has a Poisson distribution with mean $\|\varphi\|_1$, and the distances of each point of the cluster from the cluster centre are independent random variables with probability distribution function $\frac{\varphi(t)}{\|\varphi\|_1}$}''. }. We establish this analogy rigorously in Appendix \ref{app: proof prop embedding}.
\\

Conversely, if rain cell arrivals are modelled via a linear Hawkes process with an exponential kernel
\begin{equation} \label{eq: forme exp}
\varphi(t) = \alpha_1 \exp(-\beta_1 t){\bf 1}_{\{t \geq 0\}},\;\;\text{with}\;\;\alpha_1, \beta_1 >0\;\;\text{and}\;\;\alpha_1/\beta_1<1,
\end{equation}
under the stability condition $ \|\varphi\|_{1} = \alpha_1/\beta_1 < 1$, first and second-order statistics of data \eqref{eq: data small scale} are indistinguishable between BL, NS and Hawkes models, up to an appropriate parametrisation. More precisely, we have the following proposition.
\begin{prop}[Exponential Hawkes, BL and NS have same second-order statistics] \label{prop: expo second-order}
Assume that the time series $(Y_k^\delta)_{k \geq 1}$ defined in \eqref{eq: data small scale} is stationary\footnote{See Appendix \ref{sec: stationary assumption} for a discussion on stationarity.}. If the rain cell arrivals is a Hawkes process with an exponential kernel of the form \eqref{eq: forme exp}, we have
\begin{equation*} \label{eq: moment 1}
\E\big[Y_i^\delta\big] = \E[I]\E[L]  \Lambda \delta,
\end{equation*}
with $\Lambda = \frac{\mu}{1-\|\varphi\|_{1}} = \frac{\mu}{1-\alpha_1/\beta_1}$.
If the $L_i$ are exponentially distributed with parameter $\lambda_L >0$, we have
\begin{equation*} \label{eq: moment 2}
\mathrm{Var}(Y_i^\delta) =
\frac{2\delta\Lambda C_1}{\beta_1-\alpha_1} \Big(1 - \frac{1-\mathrm{e}^{-(\beta_1-\alpha_1)\delta}}{(\beta_1-\alpha_1)\delta}\Big) + \frac{2\delta \Lambda  C_2}{\lambda_L} \Big(1 - \frac{1-\mathrm{e}^{-\lambda_L\delta}}{\lambda_L\delta}\Big),
\end{equation*}
with
$$C_1 = \E[I]^2\frac{\alpha_1(2\beta_1 - \alpha_1)}{2(\beta_1-\alpha_1)(\lambda_L^2 - (\beta_1-\alpha_1)^2)},\;C_2 = \lambda_L^{-1}(\E[I^2] - C_1(\beta_1-\alpha_1)),$$
and for $i \neq j$:
\begin{equation*} \label{eq: moment 3}\mathrm{Cov}\big(Y_i^\delta, Y_j^\delta\big) = \Lambda C_1 \frac{\mathrm{e}^{-(\beta_1-\alpha_1)(|j-i|-1)\delta}}{(\beta_1-\alpha_1)^2}\big(1 - \mathrm{e}^{-(\beta_1-\alpha_1)\delta} \big)^2 +  \Lambda C_2 \frac{\mathrm{e}^{-\lambda_L(|j-i|-1)\delta}}{\lambda_L^2}\big(1 - \mathrm{e}^{-\lambda_L\delta} \big)^2.
\end{equation*}
Moreover, if the rain cell arrival process follow a BL process with parameters $(\mu_{BL}, \gamma_{BL}, \nu_{BL})$, the same formulas  hold true for the parametrisation
$$\gamma_{BL} = \beta_1-\alpha_1,\;\;\nu_{BL} = \frac{\alpha_1(1 - \alpha_1/(2\beta_1))}{1 - \alpha_1/\beta_1}\;\;\text{and}\;\;\mu_{BL}\Big(1+\frac{\nu_{BL}}{\gamma_{BL}}\Big)  = \frac{\mu}{1-\alpha_1/\beta_1}.$$
For the NS model, if $C$ follows a Poisson distribution with parameter $\nu_{NS}$ and $ f_{NS}$ is an exponential density with parameter $\gamma_{NS}$, the same formulas hold true for the parametrisation
$$\gamma_{NS} = \beta_1-\alpha_1,\;\; \nu_{NS} = \frac{\alpha_1}{\beta_1}\Big(2 -  \frac{\alpha_1}{\beta_1}\Big)\Big(1 -  \frac{\alpha_1}{\beta_1}\Big)^{-2}\;\;\text{and}\;\;\mu_{NS}\nu_{NS} = \frac{\mu}{1-\alpha_1/\beta_1}.$$ 
\end{prop}
The proof of Proposition \ref{prop: expo second-order} is given in Appendix \ref{app: proof prop expo second-order}. As an important consequence, Proposition \ref{prop: expo second-order} establishes that any conclusion drawn from an empirical method based on first and second-order statistics in the BL or NS models can equivalently be derived from a Hawkes process approach with an exponential kernel.

\subsubsection*{Critical and heavy-tailed Hawkes processes} 

Hawkes processes modelling rain cell arrivals can be connected in a unique way to large time scale via a renormalisation argument combined with a notion of criticality for heavy-tailed reproduction kernels. We consider a Hawkes process over a time interval $[0,T]$ and we are interested in its behaviour as $T$ becomes large. We allow its parameters to depend on $T$ in order to obtain non-trivial large time behaviour up to appropriate normalisation.

\begin{defi}[Critical and heavy-tailed Hawkes processes] \label{def: Critical linear Hawkes processes}
A linear Hawkes process $(N_t)_{t \in [0,T]}$ with intensity of the form \eqref{eq: def intensity Hawkes} and kernel  $\varphi(t) = \gamma(1+\beta t)^{-(1+\alpha)}$, with $\gamma, \beta >0$ that may depend on $T$ and $\alpha \in (0,1)$ independent of $T$ is asymptotically critical and heavy-tailed if 
$$\|\varphi\|_{1} = a_T < 1,$$
with
\begin{equation} \label{eq: kernel expansion}
a_T = 1-c_1T^{-\alpha}+o(T^{-\alpha}),
\end{equation}
for some $c_1>0$ (independent of $T$) and 
\begin{equation} \label{eq: baseline expansion}
\mu^{T} = c_2T^{\alpha-1}+o(T^{\alpha-1}),
\end{equation}
for some $c_2 >0$ (independent of $T$).
\end{defi}
As mentioned in the introduction, a critical and heavy-tailed Hawkes process in the sense of Definition \ref{def: Critical linear Hawkes processes} is the gateway to a fractional process behaviour under a scaling limit. It has a macroscopic rough scaling limit in large time scales that contains a trace of the criticality and the heavy-tail index, as developed in Section \ref{sec: connecting} and empirically confirmed by the macroscopic study in Section \ref{sec: rain is rough at large scales}.\\

In the rest of the section, we develop new statistical tools for Hawkes based rainfall models at small time scales; we then analyse our small time scale data.  


\subsection{Two inference methods at small time scales} \label{sec: inference small scales} \label{sec: stat methodo micro}

We develop in this section two different inference methods for estimating $(\mu, \varphi)$ in a linear Hawkes process based on indirect measurements of accumulated rainfall
$$Y_1^\delta, Y_2^\delta, \ldots, Y_k^\delta, \ldots$$
over the time horizon $[0,T_{\text{micro}}]$. The sample size is $n = \lfloor T_{\text{micro}}/\delta \rfloor$.
We model our data by
\begin{equation} \label{eq: data micro}
Y_k^\delta = \int_{(k-1)\delta}^{k\delta}Y_s\,ds,\;\;k = 1,\ldots, n= \lfloor T_{\text{micro}}/\delta \rfloor,
\end{equation}
with $(Y_t)_{t \geq 0}$ as in  \eqref{eq:intensity_rain} with a Hawkes process as rain cell arrivals generator. The first method is based on spectral estimation, see, among others, \cite{Chandler1997} for BL and NS models. Our extension to Hawkes processes governing rain cells arrivals is new and builds on the recent paper of \cite{CheyssonLang2022}. The second method is based on a penalised moment approach across scales. It is classical for parameter estimation in rainfall models, where variants of the moment formulations can be used. Here we take the same contrast as in \cite{wei2024}. This is the historical method of \cite{RodriguezIturbe1987}. However, it is usually implemented for a number of time scales that is much smaller than the one we take\footnote{Usually, 5 minutes, 1 hour, 6 hours, and 24 hours. In~\cite{RodriguezIturbe1987}, only the 6-hour and 12-hour aggregation scales are used.}.\\

\medskip
In the following, we estimate $\|\varphi\|_{1}$ and $\alpha$ in the model $(\mu, \varphi)$ from data \eqref{eq: data micro}, with 
\begin{equation} \label{eq: power-law}
\varphi(t) = \frac{\gamma}{(1+\beta t)^{\alpha+1}}.
\end{equation}
We approximate the power-law kernel \eqref{eq: power-law} by a sum of exponential functions with power-law weights
$$\varphi(t) = \sum_{i = 1}^d \alpha_i \exp(-\beta_i t),\;\;\alpha_i,\beta_i >0,\;\; \sum_{i = 1}^d \alpha_i/\beta_i < 1,$$
as in \cite{BochudChallet2007, HardimanBercotBouchaud2013}. More specifically, we take
\begin{equation} \label{eq: power-law approximation}
\varphi(t) = \frac{\eta}{Z}\sum_{i = 0}^{M-1} \frac{1}{(\tau m^i)^{1+\alpha}}\exp\Big(-\frac{t}{\tau m^i}\Big),
\end{equation}
where $Z$ is a normalizing constant such that $\|\varphi\|_{1}=\eta$ and $\alpha$ plays the same role as in \eqref{eq: power-law}. Moreover, we assume that the $(I_i)_{i \geq 1}$ and the $(L_i)_{i \geq 1}$ are exponentially distributed, with respective parameters $\lambda_I$ and $\lambda_L$. This choice is common in the literature, see for instance \cite{RodriguezIturbe1987}. The model thus has eight parameters: the baseline $\mu$, the kernel parameters $(\alpha, \eta, \tau, m, M)$, and the rain cell parameters $(\lambda_I, \lambda_L)$.

\medskip
For both methods, we further need to consider a stationary version of $(N_t)_{t \geq 0}$, see Appendix \ref{sec: stationary assumption} for a rigorous clarification of this notion.

\subsubsection*{A spectral inference method}

In turn, the process $(Y_t)_{t \geq 0}$ has a stationary version $(Y_t)_{t \in \R}$ with Fourier transform defined via
$$\mathcal F_{Y}(\omega) = \int_{\R} \mathrm{Cov}(Y_0, Y_t)\exp(-\iota \omega t)dt$$
with $\iota^2=-1$. Furthermore, define $\mathcal F(\varphi)(\omega) = \int_{\R}\varphi(t)\exp(-\iota\omega t)dt$.  We have the following proposition.
 \begin{prop} \label{prop: fourier spectral}
Let $(Y_t)_{t \in \R}$ be a stationary version of \eqref{eq:intensity_rain}, as constructed in Appendix \ref{sec: stationary assumption}, Equation \eqref{eq: def models spectral stat}, driven by a stationary linear Hawkes process governing rain-cell arrivals and such that $\E[L] < \infty$ and  $\E[I^2]< \infty$. We have
\begin{align*}
\mathcal F_{Y}(\omega) & =2\Lambda\E[I^2]\frac{1 - \mathrm{Re}(\E[\exp(\iota \omega L)])}{\omega^2}\\
&+\E[I]^2\Lambda\frac{ |1 -\E[\exp(\iota\omega L)]|^2}{\omega^2}\Big(\frac{1}{|1 - \mathcal F(\varphi)(\omega)|^2}-1\Big),
\end{align*}
with
$\Lambda = \mu(1-\|\varphi\|_1)^{-1}$. 
\end{prop}
The proof is given in Appendix \ref{app: proof prop fourier}, where we also give analogous results for the BL and NS models, see also \cite{Chandler1997}. The aggregated stationary continuous time process $(Y_t^\delta)_{t \in \R}$, defined by $Y_t^\delta  = \int_{t\delta}^{(t+1)\delta} Y_s ds$, which coincides with the discrete time process $(Y_k^\delta)_{k=1,\ldots, n}$ at integer times $t=k-1$, has Fourier transform
$$\mathcal F_{Y}^{\delta}(\omega) = \delta \mathcal F_Y(\delta^{-1}\omega)\left(\frac{\sin(\omega / 2)}{\omega / 2}\right)^2.$$

We use the parametrization 
\begin{equation} \label{eq: parametrisation}
\vartheta  = \Big(\alpha, \eta, \tau, m, M, \Lambda \lambda_I^{-2}= \frac{\mu \lambda_I^{-2}}{1-\|\varphi\|_{1}}, \lambda_L\Big),
\end{equation}
to avoid identifiability issues, as will become transparent below.\\

From the explicit representation provided by Proposition \ref{prop: fourier spectral}, we construct an estimator $\widehat \vartheta_n$ by minimising the spectral score 
\begin{equation} \label{eq:spectral_estim}
    \vartheta \mapsto \mathcal S_n(\vartheta) = \frac{1}{4\pi} \int_{-\pi}^{\pi} \left(\log f_{\vartheta}(\omega) + \frac{I_n(\omega)}{f_{\vartheta}(\omega)}\right)d\omega, 
\end{equation}
where $I_n(\omega)$ denotes the periodogram of $(Y^{\delta}_k)_{k=1,\ldots, n}$ defined by
\[I_n(\omega) = (2\pi n)^{-1} \big|\sum_{k=1}^n \big(Y^{\delta}_k-\overline{Y}_n^\delta\big) e^{-ik\omega}\big|^2,\;\; \overline{Y}_n^\delta = \frac{1}{n}\sum_{k=1}^n Y^{\delta}_k,\] 
and
$$f_{\vartheta}(\omega) = \sum_{k \in \mathbb{Z}} \mathcal F_{Y}^{\delta}(\omega + 2k\pi).$$
Since $L$ is exponentially distributed with parameter $\lambda_L$, we have $\mathbb{E}[L^2] = 2\mathbb{E}[L]^2 = 2\lambda_L^{-2}$ and  
\[\frac{1 - \mathrm{Re}(\E[\exp(i\omega L)])}{\omega^2} =\frac{|1 -\E[\exp(i\omega L)]|^2}{\omega^2} = \frac{1}{\lambda_L^2+  \omega^2}.
\]
The spectral density identifies only the product $\Lambda\mathbb E[I^2]$, or equivalently $\Lambda\mathbb E[I]^2$ under the exponential assumption on $I$, which leads to an identifiability issue. To overcome this, we estimate $\Lambda \mathbb{E}[I]^2 = \Lambda \lambda_I^{-2}$ in the optimisation problem~\eqref{eq:spectral_estim} using the parametrisation~\eqref{eq: parametrisation}. Combining the estimate of $\Lambda\lambda_I^{-2}$ with the additional moment equation
\[
\mathbb{E}[Y_k^{\delta}]
=
\frac{\Lambda\delta}{\lambda_I\lambda_L}
\]
allows us to estimate $\Lambda$ and $\lambda_I$ separately. We then recover
\[
\mu=\Lambda\bigl(1-\|\varphi\|_1\bigr).
\]

\subsubsection*{A second-order contrast method across scales}

 We have data $(Y_k^\delta)_{k=1,\ldots, n}$ from \eqref{eq: data micro}.
The parameters of the model are thus 
\begin{equation} \label{eq: parameter multiscale}
\vartheta = \big(\alpha, \eta, \tau, m, M, \mu, \lambda_I, \lambda_L\big).
\end{equation}
We have explicit forms of second-order statistics of the time series $(Y_k^\delta)_{k=1,\ldots, n}$ when $\varphi$ is parametrised by sum of exponential functions with weights mimicking a power-law. In turn, we obtain simple and explicit contrast functions with several numerical and stability advantages over the spectral method.\\

More precisely, let $\alpha_i, \beta_i >0$,  $i=1,\ldots, d$, be parameters satisfying $\sum_{i = 1}^d \frac{\alpha_i}{\beta_i} < 1$ with all $\beta_i$ distinct. Let $p_1,\ldots,p_d>0$ denote the distinct roots of the polynomial
$$P(\omega) = \prod_{i=1}^d (\beta_i-\omega) - \sum_{j=1}^d \alpha_j \prod_{\substack{i=1\\i\neq j}}^d (\beta_i-\omega),$$
$$A_i =
 -\frac{\prod_{j=1}^d (\beta_j^2-p_i^2)}{P(-p_i)P'(p_i)},$$
where $P'$ denotes the derivative of $P$, and
$$C_i = \frac{\mathbb{E}[I]^2 A_i}{\lambda_L^2-p_i^2},\;i = 1,\ldots, d,\;\;C_{d+1} = \lambda_L^{-1}\Big(\mathbb{E}[I^2] -\mathbb{E}[I]^2  \sum_{i=1}^d \frac{A_ip_i}{\lambda_L^2-p_i^2}\Big),$$
with $\Lambda = \frac{\mu}{1-\|\varphi\|_{1}}=
\frac{\mu}{1-\sum_{i=1}^d\alpha_i/\beta_i}.$ We have the following proposition.
\begin{prop} \label{prop: multiscale contrast}
Let $(Y_t)_{t \in \R}$ be a stationary version of \eqref{eq:intensity_rain}, as constructed in Appendix \ref{sec: stationary assumption}, Equation \eqref{eq: def models spectral stat}, driven by a stationary linear Hawkes process governing rain-cell arrivals with kernel $\varphi(t) = \sum_{i=1}^d \alpha_i e^{-\beta_i t}$, with $\alpha_i>0$ and the $\beta_i > 0$ all different, satisfying the stability condition $\sum_{i=1}^d \frac{\alpha_i}{\beta_i} < 1$. Assume further that the rain-cell durations $(L_i)_{i\geq1}$ are exponentially distributed with parameter $\lambda_L>0$, $\lambda_L \neq p_i$ for $i=1,\ldots,d$, and that $\mathbb E[I^2]<\infty$. For arbitrary $h >0$ and integer $k \geq 1$, let $Y_k^h = \int_{(k-1)h}^{kh}Y_s\,ds$. We have 
\[
\mathrm{Var}(Y_k^{h})=
2h \Lambda \sum_{i=1}^d \frac{C_i}{p_i}\Big(1 - \frac{1-\exp(-p_ih)}{p_ih}\Big)+\frac{2h \Lambda  C_{d+1}}{\lambda_L} \Big(1 - \frac{1-\exp(-\lambda_Lh)}{\lambda_Lh}\Big), 
\]
and for $|k-k'| \geq 1$, we have
\begin{align*}
\mathrm{Cov}(Y_k^{h}, Y_{k'}^{h}) & =  \Lambda \sum_{i=1}^d \frac{C_i}{p_i^2} \exp(-p_i(|k-k'|-1)h)\big(1 - \exp(-p_i h)\big)^2
\\
&+ \frac{ \Lambda C_{d+1}}{\lambda_L^2} \exp(-\lambda_L(|k-k'|-1)h)\big(1 - \exp(-\lambda_Lh) \big)^2.
\end{align*}
\end{prop}
The proof is given in Appendix \ref{app: proof of multiscale contrast}. The remarkable simplicity of Proposition \ref{prop: multiscale contrast} enables us to build a contrast across scales based on second-order statistics by minimizing
\begin{align}
\vartheta \mapsto \mathcal L_n(\vartheta) = 
&\sum_{h \in \mathcal{H}} \Big(w_{1}(h)\big(E_n[Y_i^h]-\mathbb{E}_{\vartheta}[Y_i^{h}] \big)^2 +w_{2}(h)\Big(\frac{V_n(Y_i^h)^{1/2}}{E_n[Y_i^h]}-\frac{\mathrm{Var}_\vartheta(Y_i^h)^{1/2}}{\E_\vartheta[Y_i^h]}\Big)^2 \nonumber \\
&+w_3(h)\Big(\frac{C_n(Y_i^h,Y_{i+1}^h)}{V_n(Y_i^h)}-\frac{\mathrm{Cov}_{\vartheta}(Y_i^h,Y_{i+1}^h)}{\mathrm{Var}_{\vartheta}(Y_i^h)}\Big)^2\Big), \label{eq: def contraste}
\end{align}
where $\mathcal H$ consists of a grid of the form $i\delta$ for the largest set of indices $i \geq 1$ that are observable thanks to the aggregation property $Y_k^h =
\sum_{l=(k-1)i+1}^{ki}Y_l^\delta$ as soon as $h=i\delta$. The $w_i(h)$ are non-negative weights\footnote{The weights are chosen with the method used in~\cite{kaczmarska2014}: for a given statistic, the associated weight is equal to the inverse of the empirical variance of the same statistic, calculated each year. In order not to penalise some timescales (typically the daily timescale where there are fewer data), we consider the same weights for all scales for a given statistic (mean, coefficient of variation, autoregressive coefficient), calculated as the mean of the different initial weights (inverse of the variance) over the different timescales.}, the operators $\E_\vartheta$, $\mathrm{Var}_\vartheta$
and $\mathrm{Cov}_{\vartheta}$ denote expectation, variance and covariances of the $Y_k^h$ computed for the parameter $\vartheta$ defined in \eqref{eq: parameter multiscale}, obtained thanks to Proposition \ref{prop: multiscale contrast}, and $E_n$, $V_n$, and $C_n$ denote their respective empirical counterparts.


\subsection{Statistical evidence of criticality and heavy-tails at small time scales} \label{sec: stat criticality}

We analyse two categories of datasets:
\begin{itemize}
\item The Bochum weather station data\footnote{Bochum is a city in the Ruhr area in Germany, Central Europe. The data were recorded from January 1931 to December 1999, using a Hellmann rain gauge. They are analysed in depth in \cite[Section 3.2]{kaczmarska2013} and are available at {\tt https://github.com/NTU-CompHydroMet-Lab/pyBL/blob/main/examples/data.zip}. The Bochum dataset is provided with the Python package pyBL~\cite{wei2024}.}. We have $\delta = 5$ minutes and $T_{\text{micro}} = 69$ years.

\item Data from four M\'et\'eo France weather stations\footnote{Available on~\url{https://meteo.data.gouv.fr/}. }:  Lille (M\'et\'eo France ID: 59343001), Marseille (M\'et\'eo France ID: 13054001), Strasbourg (M\'et\'eo France ID: 67124001), and Toulouse (M\'et\'eo France ID: 31069001). We have $\delta = 6$ minutes and $T_{\text{micro}} = 20$ years.
  \end{itemize}
 The coordinates of the weather stations are reported in Table~\ref{tab:median_MF}. \\
  

We implement both the spectral method and the second-order contrast method of Section \ref{sec: stat methodo micro} to estimate $\|\varphi\|_{1}$ and $\alpha$ for each month over 69 years, thereby accounting for basic seasonality. We build 95\% Monte-Carlo confidence intervals and estimate the standard deviations of the estimators from 100 repeated samples of data generated with the estimated parameters. We also compute an estimate of the statistical information $\mathbb{E}[N_T] \sim \frac{\mu T}{1-\|\varphi\|_{1}}$ for each combined month, where $\mu$ and $\|\varphi\|_{1}$ are replaced by our estimators. The results are displayed in 
Table \ref{table: second-order} for the second-order contrast method and in Table \ref{table: sub spectral} for the spectral approach.\\

\begin{table}[H]
\centering
\begin{tabular}{lllllll}
\toprule
Month & $\|\varphi\|_1$ & $\sigma_{\|\varphi\|_1}$ & $\alpha$ & $\sigma_{\alpha}$ & $n_{>0}$ & $\frac{\mu T}{1-\|\varphi\|_1}$ \\
\midrule
Jan & 0.971 & 0.004 & 0.531 & 0.042 & \num{6.50e+04} & \num{2.73e+05} \\
Feb & 0.930 & 0.007 & 0.682 & 0.062 & \num{5.24e+04} & \num{5.56e+04} \\
Mar & 0.928 & 0.008 & 0.578 & 0.042 & \num{5.25e+04} & \num{4.89e+04} \\
\midrule
Apr & 0.931 & 0.008 & 0.614 & 0.053 & \num{4.83e+04} & \num{6.30e+04} \\
May & 0.872 & 0.012 & 0.764 & 0.083 & \num{4.08e+04} & \num{2.56e+04} \\
Jun & 0.951 & 0.012 & 0.422 & 0.049 & \num{4.07e+04} & \num{1.54e+05} \\
\midrule
Jul & 0.776 & 0.035 & 0.678 & 0.841 & \num{3.89e+04} & \num{7.96e+03} \\
Aug & 0.842 & 0.016 & 0.407 & 0.053 & \num{3.49e+04} & \num{9.67e+03} \\
Sep & 0.856 & 0.012 & 0.646 & 0.074 & \num{3.78e+04} & \num{1.48e+04} \\
\midrule
Oct & 0.928 & 0.006 & 0.640 & 0.038 & \num{4.77e+04} & \num{4.65e+04} \\
Nov & 0.940 & 0.009 & 0.607 & 0.057 & \num{6.24e+04} & \num{6.25e+04} \\
Dec & 0.978 & 0.003 & 0.528 & 0.040 & \num{6.60e+04} & \num{4.03e+05} \\
\bottomrule
\end{tabular}
\caption{{\it {\small Parameter estimates for $\alpha$ and $\|\varphi\|_{1}$ for the Bochum station under the Hawkes model with the power-law approximation \eqref{eq: power-law approximation} based on the second-order contrast method, with standard deviation 
$\sigma_{\|\varphi\|_{1}}$ or $\sigma_\alpha$ based 
on 100 repeated simulations with estimated parameters. The number $n_{>0}$ indicates the number of non-zero data. The last column displays a proxy of the statistical information (in number of events) $\frac{\mu T}{1-\|\varphi\|_1}$, where $\mu$ and $\|\varphi\|_{1}$ are replaced by our estimators.}}}
 \label{table: second-order}
\end{table}

\begin{table}[H]
\centering
\begin{tabular}{lllllll}
\toprule
Month & $\|\varphi\|_1$ & $\sigma_{\|\varphi\|_1}$ & $\alpha$ & $\sigma_{\alpha}$ & $n_{>0}$ & $\frac{\mu T}{1-\|\varphi\|_1}$ \\
\midrule
Jan & 0.989 & 0.003 & 0.657 & 0.026 & \num{6.17e+04} & \num{3.98e+06} \\
Feb & 0.992 & 0.002 & 0.671 & 0.024 & \num{4.85e+04} & \num{8.15e+06} \\
Mar & 0.984 & 0.004 & 0.557 & 0.030 & \num{4.89e+04} & \num{1.67e+06} \\
\midrule
Apr & 0.978 & 0.004 & 0.506 & 0.031 & \num{4.65e+04} & \num{9.02e+05} \\
May & 0.903 & 0.013 & 0.751 & 0.056 & \num{3.93e+04} & \num{5.17e+04} \\
Jun & 0.961 & 0.013 & 0.622 & 0.065 & \num{3.83e+04} & \num{3.70e+05} \\
\midrule
Jul & 0.969 & 0.012 & 0.780 & 0.082 & \num{3.61e+04} & \num{7.37e+05} \\
{\color{red}Aug} & 0.902 & 0.007 & 1.535 & 1.156 & \num{3.37e+04} & \num{7.98e+04} \\
Sep & 0.976 & 0.008 & 0.628 & 0.048 & \num{3.65e+04} & \num{9.67e+05} \\
\midrule
Oct & 0.970 & 0.005 & 0.509 & 0.024 & \num{4.49e+04} & \num{2.52e+05} \\
Nov & 0.989 & 0.003 & 0.599 & 0.025 & \num{6.01e+04} & \num{3.47e+06} \\
Dec & 0.988 & 0.002 & 0.634 & 0.021 & \num{5.82e+04} & \num{2.98e+06} \\
\bottomrule
\end{tabular}\caption{{\it {\small Same experiment as in Table \ref{table: second-order}, except that the spectral method is used for parameter estimation instead of the second-order contrast method. Months shown in red indicate cases where the power-law Hawkes model does not achieve the lowest AIC score~\eqref{eq:aic} among the Hawkes model with an exponential kernel, the BL model, and the NS model.}}}
 \label{table: sub spectral}
\end{table}

A comparative analysis of Table \ref{table: second-order} and \ref{table: sub spectral} shows that, although both approaches yield consistent results, the spectral method tends to estimate $\|\varphi\|_{1}$ closer to one than the second-order contrast method. Also, it seems that criticality is statistically more pronounced for rainier periods of the year, as measured by the number of non-zero observations $n_{>0}$, which is of course not a surprise. From a physical standpoint, \cite{marani2003} explains that events are less correlated in summer, when precipitation is mainly associated with isolated convective phenomena, than in winter, when it is linked to synoptic-scale systems. The detailed results for the four M\'et\'eo France data are provided in Appendix~\ref{app:mf_micro_results} and lead to the same conclusions. Here, we only provide the median values for $\|\varphi\|_1$ and $\alpha$ across months for each station in Table~\ref{tab:median_MF}. 

\begin{table}[H]
\centering
\begin{tabular}{lllcccc}
\toprule
City &Lon &Lat & \multicolumn{2}{c}{Contrast} & \multicolumn{2}{c}{Spectral} \\
\cmidrule(lr){4-5} \cmidrule(lr){6-7}
& & & $\|\varphi\|_1$ & $\alpha$ & $\|\varphi\|_1$ & $\alpha$ \\
\midrule
Bochum & 7.2$^\circ$E & 51.5$^\circ$N & 0.929 & 0.611 & 0.978 & 0.628 \\
Lille & 3.1$^\circ$E & 50.6$^\circ$N & 0.926 & 0.484 & 0.923 & 0.877 \\
Marseille & 5.2$^\circ$E & 43.4$^\circ$N & 0.932 & 0.581 & 0.955 & 0.800 \\
\midrule
Strasbourg & 7.6$^\circ$E & 48.5$^\circ$N & 0.930 & 0.571 & 0.915 & 0.806 \\
Toulouse & 1.4$^\circ$E & 43.6$^\circ$N & 0.905 & 0.516 & 0.936 & 0.707 \\
\bottomrule
\end{tabular}
\caption{\label{tab:median_MF}\small \it Median across months of the estimates of $\|\varphi\|_1$ and $\alpha$ under Hawkes model with the power-law approximation \eqref{eq: power-law approximation}, obtained using the second-order contrast method and the spectral approach. For the second-order contrast method, we discard the very few months for which the considered model does not attain the lowest score in \eqref{eq:contrast} among the Hawkes model with an exponential kernel, the BL model, and the NS model. Similarly, for the spectral method, we discard the very few months for which the model does not attain the lowest AIC score in \eqref{eq:aic}.
}
\end{table}
We next compare the goodness-of-fit of several models in Figure~\ref{fig:score_moment} and Figure~\ref{fig:score_spectral}, for the following set of models: a Hawkes with an exponential kernel, a Hawkes with a power-law kernel approximated by a sum of multiscale exponentials, Bartlett-Lewis and Neyman-Scott (the last two viewed as benchmarks). We compute a goodness-of-fit criterion for each estimation method. For the spectral method, we define an AIC score by setting 
\begin{equation}
    \label{eq:aic}
2n \mathcal S_n(\widehat \vartheta_n)+2p,
\end{equation}
where $\widehat \vartheta_n$ minimises $\vartheta \mapsto \mathcal S_n(\vartheta)$ defined in \eqref{eq:spectral_estim} and $p$ is the dimension of the model: $p=5$ for the exponential Hawkes, Bartlett-Lewis, and Neyman-Scott models, while $p=8$ for the Hawkes model with a power-law kernel approximated by a sum of exponentials \eqref{eq: power-law approximation}. For the second-order contrast method, our goodness-of-fit criterion is simply
\begin{equation} \label{eq:contrast}
\mathcal L_n(\widehat \vartheta_n),
\end{equation}
where $\widehat \vartheta_n$ minimises  $\vartheta \mapsto \mathcal L_n(\vartheta)$ defined in \eqref{eq: def contraste}. Whereas the AIC is an information criterion that accounts for model complexity by penalising the number of parameters, thereby enabling a direct statistical comparison of competing models, the criterion for the second-order contrast method only measures the gain in model fit resulting from the inclusion of additional parameters. Because both criteria rely exclusively on first and second-order statistics, the Hawkes model with a simple exponential kernel should, in theory, achieve the same performance as the BL and NS models (see Proposition \ref{prop: expo second-order}). Consequently, any differences observed in practice can only be attributed to numerical optimisation effects.\\



\begin{figure}[H]
\begin{tabular}{ccc}
        \includegraphics[width=0.3\linewidth]{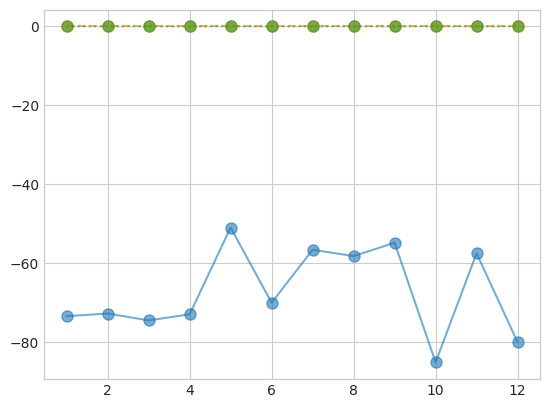}  &    \includegraphics[width=0.3\linewidth]{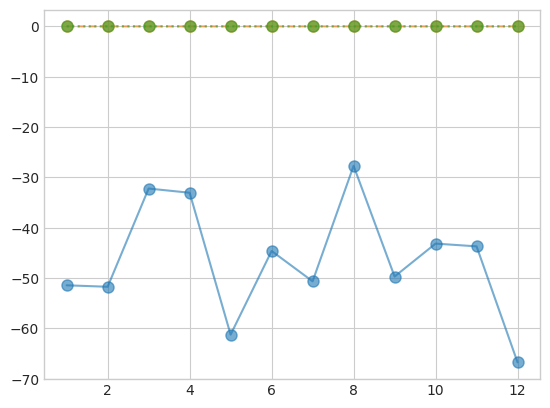}  &    \includegraphics[width=0.3\linewidth]{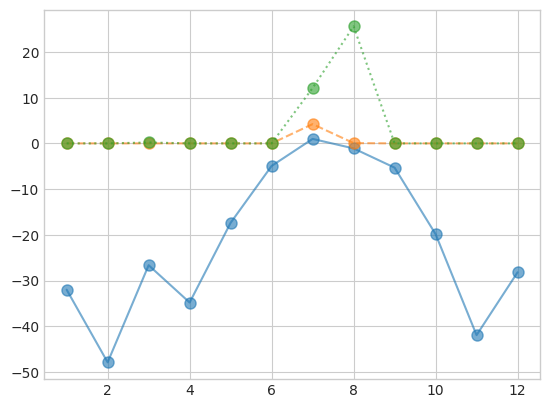} \\
   Bochum   &  Lille & Marseille \\
               \end{tabular}
               \centering
               \begin{tabular}{cccc}
     
               \includegraphics[width=0.3\linewidth]{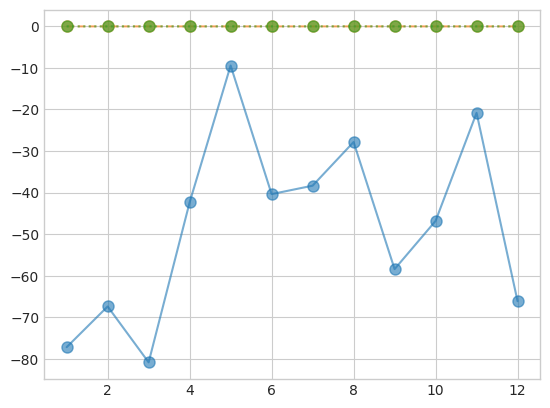}  &   &  \includegraphics[width=0.3\linewidth]{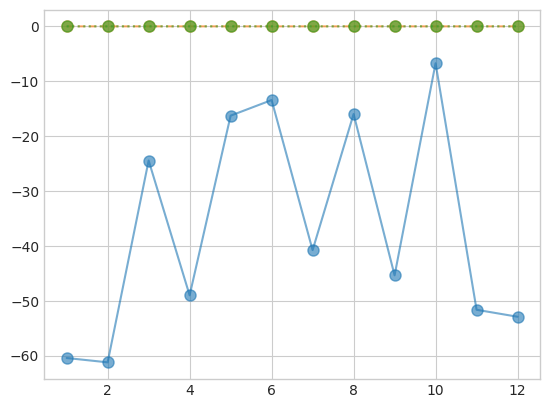} &\\
             Strasbourg & &Toulouse &
\end{tabular}
\caption{{\small {\it  Relative difference (monthly, \%) of the scores in the second-order contrast method across scales between:  Hawkes model with a power-law approximation kernel \eqref{eq: power-law approximation}, Bartlett-Lewis model, Neyman-Scott model, and the score obtained by the Hawkes fits for an exponential kernel for different weather stations. Blue: Hawkes with power-law approximation kernel \eqref{eq: power-law approximation}, Orange: Bartlett-Lewis, Green: Neyman-Scott. The Hawkes model with a simple exponential kernel has the same performance as BL or NS in theory, by Proposition \ref{prop: expo second-order}; any observed differences are therefore only due to numerical optimisation. The relative difference has been chosen to provide more readable results, since the values of $\mathcal{L}_n(\widehat \vartheta_n)$ can vary significantly from one month to another. This is not the case for the AIC score shown in Figure~\ref{fig:score_spectral}.}}}
    \label{fig:score_moment}
\end{figure}

\begin{figure}[H]
\begin{tabular}{ccc}
        \includegraphics[width=0.3\linewidth]{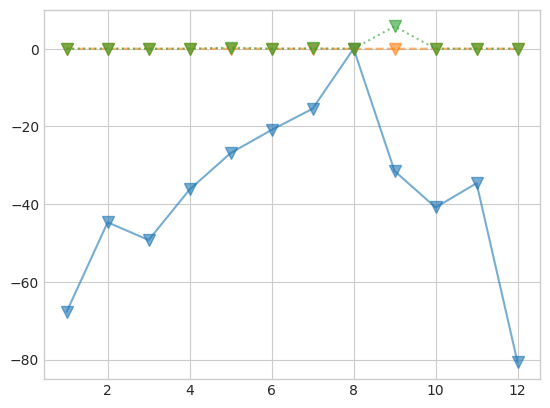}  &    \includegraphics[width=0.3\linewidth]{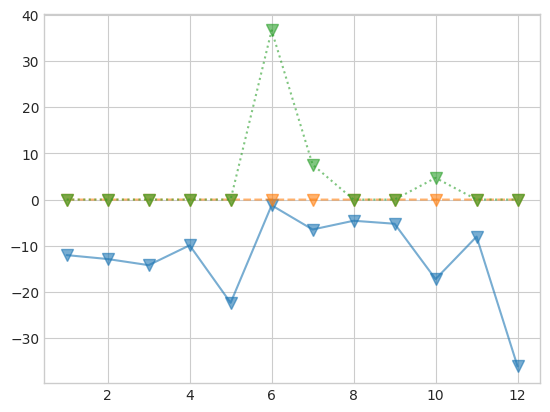}  &    \includegraphics[width=0.3\linewidth]{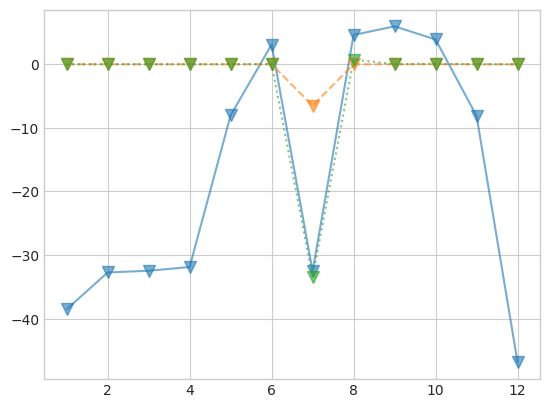} \\
   Bochum   &  Lille & Marseille \\
               \end{tabular}
               \centering
               \begin{tabular}{cccc}
     
               \includegraphics[width=0.3\linewidth]{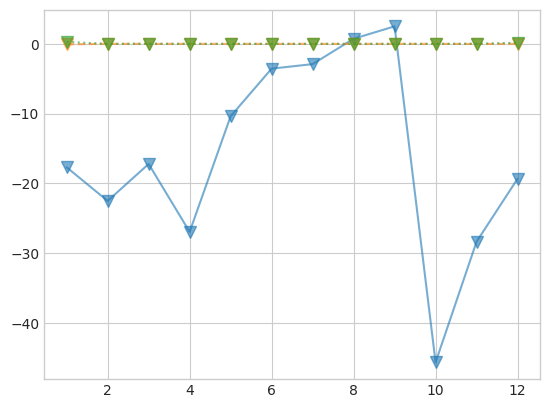}  &   &  \includegraphics[width=0.3\linewidth]{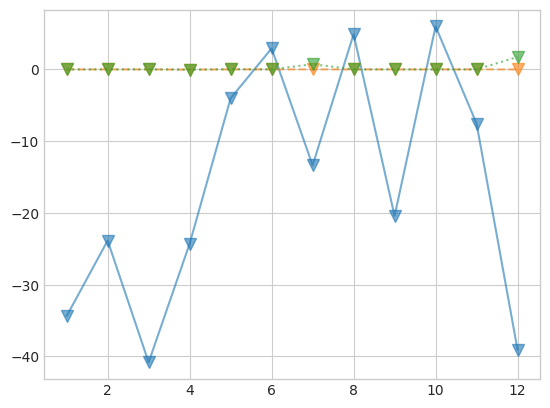} &\\
             Strasbourg & &Toulouse &
\end{tabular}
    \caption{{\small {\it Difference (monthly) of spectral AIC scores between: Hawkes model for the power-law approximation \eqref{eq: power-law approximation}, Bartlett-Lewis model, Neyman-Scott model model, and the AIC obtained by the Hawkes fits for an exponential kernel for different weather stations. Blue: Hawkes with power-law approximation kernel \eqref{eq: power-law approximation}, Orange: Bartlett-Lewis, Green: Neyman-Scott. The Hawkes model with a simple exponential kernel has the same performance as BL or NS in theory, by Proposition \ref{prop: expo second-order}; any observed differences are therefore only due to numerical optimisation.}}}
    \label{fig:score_spectral}
\end{figure}


The results consistently show that Hawkes models with a power-law approximation~\eqref{eq: power-law approximation} generally provide the best statistical fits\footnote{With respect to classical Neyman-Scott, Bartlett-Lewis, and the Hawkes model with an exponential kernel as far as second-order statistics are concerned.} for modelling arrivals of rain cells. They systematically achieve the best performance under the second-order contrast method (except for Marseille in June) and, in most cases, also minimise the AIC score under the spectral method. Thus, they significantly improve the representation of the second-order moments of the data while introducing only a few additional parameters and a limited increase in model complexity. \\

Heavy-tailed kernels $\varphi(t) \approx c_1(1+c_2t)^{-(1+\alpha)}$ provide the best statistical fits, with criticality, in the sense that $\|\varphi\|_{1} \approx 1$ and $\alpha \approx 0.5+$. Of course, the symbol $\approx$ has to be taken with some care, but our evidence is sufficiently strong to pursue our study from small to large time scales with the roughness paradigm. 


    

\subsection{Empirical validation of the asymptotic regime} 

As mentioned in Section \ref{sec: intro from crit to rough} and developed in detail in Section \ref{sec: connecting} below, a critical Hawkes process with a heavy-tailed kernel is paramount to obtain a rough limit when scaling the process in large time. A rough limit can only be obtained in certain regimes for $\|\varphi \|_{1}$ and $\mu$ as functions of $T$, see Definition \ref{def: Critical linear Hawkes processes}. More precisely, if such a limit prevails, by \eqref{eq: kernel expansion} and \eqref{eq: baseline expansion} we must have
$$\mu = cT^{\alpha-1}+o(T^{\alpha-1})\;\;\text{and}\;\;\|\varphi\|_{1}  = 1-c'T^{-\alpha}+o(T^{-\alpha})\;\;\text{as}\;T\rightarrow \infty.$$
Moreover, since $N_T$ is of order $\frac{\mu T}{1-\|\varphi\|_{1}}$, we must have
$$N_T \sim \frac{\mu T}{1-\|\varphi\|_{1}} \sim \frac{cT^{\alpha-1}T}{c'T^{-\alpha}} \sim c'' T^{2\alpha}.$$
In turn,
$$\log \frac{\mu T}{1-\|\varphi\|_{1}} \sim 2\alpha \log T,$$
and therefore $\alpha$ must be close to $\log  \frac{\mu T}{1-\|\varphi\|_{1}}/(2\log (T))$ when $T$ is large.\\

To validate our asymptotic setting, we can test this asymptotic identity, at least in term of order of magnitude, using our estimators of $\mu$, $\alpha$, and $\|\varphi\|_{1}$. We display the results in Figures \ref{fig:link_parameter_bochum} and \ref{fig:link_parameter} for the Bochum data and the M\'et\'eo France data. We obtain on data that $\log(\frac{\mu T}{1 - \|\varphi\|_1})/(2\log(T))$ indeed lies essentially between $0.5\alpha$ and $1.5\alpha$, which is another indication of the relevance of our approach.

\begin{figure}[H]
    \centering
    \includegraphics[width=0.35\linewidth]{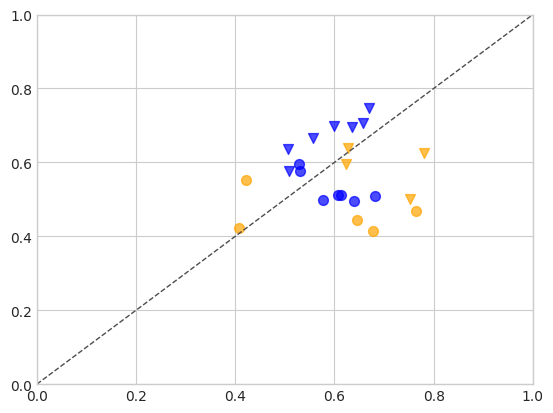}
    \caption{{\small {\it $\log(\frac{\mu T}{1 - \|\varphi\|_1})/(2\log(T))$ as a function of $\alpha$ in the Hawkes multiscale model with parameters estimated on Bochum dataset, month by month, with the moment method (o) and the spectral method ($\triangledown$). Blue: winter months (October to March), Orange: summer months (April to September). For the second-order contrast method, we discard the very few months for which the considered model does not attain the lowest score in \eqref{eq:contrast} among the Hawkes model with an exponential kernel, the BL model, and the NS model. Similarly, for the spectral method, we discard the very few months for which the model does not attain the lowest AIC score in \eqref{eq:aic}.}}}
    \label{fig:link_parameter_bochum}
\end{figure}

\begin{figure}[H]
    \centering
    \begin{tabular}{cc}
        \includegraphics[width=0.35\linewidth]{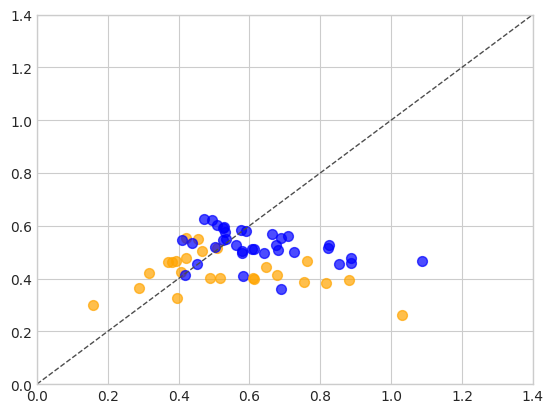} &     \includegraphics[width=0.35\linewidth]{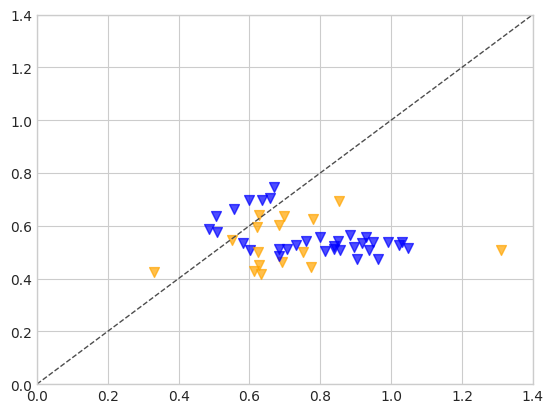}
    \end{tabular}
    \caption{{\small {\it $\log(\frac{\mu T}{1 - \|\varphi\|_1})/(2\log(T))$ as a function of $\alpha$ in the Hawkes multiscale model with parameters estimated on the different dataset, month by month, with the moment method (left) or the spectral method (right) for the five datasets. Blue: winter months (October to March), Orange: summer months (April to September). For the second-order contrast method, we discard the very few months for which the considered model does not attain the lowest score in \eqref{eq:contrast} among the Hawkes model with an exponential kernel, the BL model, and the NS model. Similarly, for the spectral method, we discard the very few months for which the model does not attain the lowest AIC score in \eqref{eq:aic}. 
    }}}
    \label{fig:link_parameter}
\end{figure}


\newpage

\section{Connecting small and large time scales} \label{sec: connecting}

\subsection{Connecting Hawkes and fractional processes} \label{sec: connecting hawkes and fractional}

In \cite{JaissonRosenbaum2015, JaissonRosenbaum2016}, the following is established:
Consider a linear Hawkes process $N_t$ with baseline $\mu >0$ and a power-law kernel 
\begin{equation} \label{eq: power-law kernel}
\varphi(t) =   \frac{\gamma}{(1+\beta t)^{1+\alpha}}{\bf 1}_{\{t \geq 0\}},
\end{equation}
with parameters $\alpha \in (1/2, 1), \beta,\gamma >0$,
as in the representation \eqref{eq: def intensity Hawkes}.
If we allow $\mu$ and $\varphi$ to depend on $T$ so that $N_t = (N_t)_{t \in [0,T]}$ is heavy-tailed and critical in the sense of Definition \ref{def: Critical linear Hawkes processes} then
\begin{equation} \label{eq: limit CLT}
T^{-2\alpha}N_{tT} \rightarrow \int_0^t X_s\,ds,\;\;t \in [0,1],
\end{equation}
as $T \rightarrow \infty$ in distribution\footnote{The convergence holds in the sense that the family $T^{-2\alpha}(N_{tT})_{t \in [0,1]}$ is tight (for the Skorokhod topology of c\`adl\`ag processes), and every limit point $\big(\int_0^tX_s\,ds\big)_{t \in [0,1]}$ is differentiable, with derivative process $(X_t)_{t \in [0,1]}$ satisfying
$$X_t = g(t)+\int_0^t k_{H}(t-s)\sqrt{X_s}dB_s,\;\;t \in [0,1],$$
where $(B_t)_{t \in [0,1]}$ is a Brownian motion, $g$ and $k_H$ are explicit functions depending only on $c_1$, $c_2$, and $H$, and $g$ is smooth, see also \cite{Horst2023HeavyTailed} for more precise results on this convergence.}, where $X_t = (X_t)_{t \in [0,1]}$ is a fractional process of the form \eqref{eq: fractional processes} below with Hurst index $H=\alpha-1/2$.\\

We are now ready to connect rainfall models at small time scales and fractional processes at large time scales via the parameters $\alpha$ and $H$. Consider a rainfall intensity model \eqref{eq:intensity_rain} at small scales
$$Y_t = \sum_{i=1}^{N_t} I_i \,{\bf 1}_{\{\tau_i + L_i > t\}} = \sum_{i \geq 1}I_i {\bf 1}_{\{\tau_i \leq t < \tau_i + L_i \}},$$
with independent and identically distributed pairs of rain-cell intensities and durations $(I_i,L_i)$, independent of the arrival process, having finite second moments, and such that $I_i$ and $L_i$ are independent. Here, $N_t$ is a heavy-tailed critical Hawkes process according to Definition \ref{def: Critical linear Hawkes processes}. With a little algebra, we obtain, for $t \in [0,T]$,
\begin{align*}
\int_0^{t} Y_s\,ds  & = \sum_{i \geq 1} I_i{\bf 1}_{\{\tau_i \leq t\}}\big(\min(t, \tau_i+L_i)-\tau_i\big) \\
& = \sum_{i \geq 1} I_iL_i{\bf 1}_{\{\tau_i \leq t\}} - \sum_{i \geq 1} I_i(\tau_i+L_i-t){\bf 1}_{\{\tau_i \leq t \leq \tau_i+L_i\}} \\
& = \sum_{i = 1}^{N_t}I_iL_i+\xi_t.
\end{align*}
Let us rescale time over $[0,1]$. The previous representation becomes
\begin{equation} \label{eq: macro approx}
\int_0^{tT} Y_s\,ds  = \kappa N_{tT}  +   \xi_{tT} + \zeta_{tT},\;\;t \in [0,1],
\end{equation}
with $\kappa = \E[I]\E[L]$, anticipating a law of large numbers induced by $N_{tT} \rightarrow \infty$ almost surely as $T\rightarrow \infty$. More precisely, we have the following lemma.
\begin{lem} \label{lemma: error bound}
Assume that $\E[L^2]<\infty$, $\E[I^2] < \infty$ and that $(N_t)_{t \geq 0}$ is a heavy-tailed critical Hawkes process according to Definition \ref{def: Critical linear Hawkes processes}. We have 
$$\sup_{t \geq 0}\E\big[|\xi_{t}|] \lesssim T^{2\alpha-1}\;\;\text{and}\;\;\sup_{t \in [0,1]}\E\big[|\zeta_{tT}|\big] \lesssim T^\alpha.$$
\end{lem}
The proof of Lemma \ref{lemma: error bound} is given in Appendix \ref{app: proof of error bound}. Multiplying both sides of \eqref{eq: macro approx} by $T^{-2\alpha}$, using Lemma \ref{lemma: error bound} and the limit theorem \eqref{eq: limit CLT}, we obtain
\begin{equation} \label{eq: serious scale}
T^{-2\alpha} \int_0^{tT} Y_sds \rightarrow \kappa \int_0^t X_s\,ds,\;\;t \in [0,1],
\end{equation}
where $X_t = (X_t)_{t \in [0,1]}$ is a fractional process in the sense of Definition \ref{def: fractional process} in Section \ref{sec: model large scale data} below with Hurst exponent $H=\alpha - 1/2$. 

\subsection{Reconciling small time scale data and large time scale data}

In turn, this result enables us to rigorously connect the data we have across two different time scales. On the one hand, we have small time aggregated rainfall data:
\begin{equation} \label{eq: data set micro}
Y_1^\delta, Y_2^\delta, \ldots,  Y_{n}^\delta\;\;\;\text{over}\;\;\;[0,T_{\text{micro}}],
\end{equation}
and $n  =  \lfloor T_{\text{micro}}/\delta\rfloor$. On the other hand, we have large time aggregated rainfall data:
\begin{equation} \label{eq: yet an embedding}
Z_1^\Delta , Z_2^\Delta, \ldots, Z_N^\Delta\;\;\;\text{over}\;\;\;[0,T_{\text{macro}}],
\end{equation}
and $N  =  \lfloor T_{\text{macro}}/\Delta\rfloor$. How do we reconcile these two datasets?  In practice, they are extracted from completely different time periods, locations, and measurement protocols. Yet, if we mathematically embed them into the same framework, we can relate a key property shared by both datasets.\\

 We define a continuous time embedding via a random process $(Z_t)_{t \in [0,1]}$ defined by
\begin{equation} \label{eq: continuous embedding}
Z_{k\Delta/T_{\text{macro}}} = Z_k^\Delta,\;\;\text{for}\;\;k=1,\ldots, N,
\end{equation}
that interpolates\footnote{As for the values of $(Z_t)_{t \in [0,1]}$ for $t \neq k\Delta/ T_{\text{macro}}$, we do not need to specify yet.} the times series $(Z_k^\Delta)_{k=1,\ldots, N}$ at discrete time $k\Delta/ T_{\text{macro}}$. \\

\begin{prop} \label{prop: small meets large}
Set $T_{\mathrm{micro}}=T$, $T_{\mathrm{macro}}=1$ and let $T \rightarrow \infty$. Assume the data sets \eqref{eq: data set micro} and \eqref{eq: yet an embedding} are compatible in the following sense:  
$$
Z_k^\Delta  = \sum_{\ell = (k-1)\Delta T/\delta+1}^{k\Delta T/\delta} Y_\ell^\delta,
$$
Then, in distribution,
$$Z_{k\Delta} = Z_k^\Delta = \int_{(k-1)\Delta T}^{k\Delta T} Y_s\,ds  = \kappa \Delta  T^{2\alpha}X_{k\Delta} +o(T^{2\alpha})$$
as $T \rightarrow \infty$.  
\end{prop} 

In particular, the macroscopic data $ Z_k^\Delta$ is well approximated (up to rescaling in space) by the process $Z_t$ which has the same smoothness properties than $X_t$, namely that of a rough process with Hurst index $H = \alpha - 1/2$.
The proof is given in Appendix \ref{sec: proof of small meets large}.

\section{Rainfall is rough at large time scales} \label{sec: rain is rough at large scales}

\subsection{Aggregated rainfall at large time scales as a fractional process} \label{sec: model large scale data}
For a fixed geographical area and a temporal aggregation $\Delta >0$, we now have measurements (direct or indirect) of aggregated rainfall, of  the form 
$$Z_1^\Delta, Z_2^\Delta, \ldots, Z_k^\Delta,\ldots.$$
By large time scales, we mean $\Delta$ ranging from 1 to 30 years. The (calendar) time period $[0, T_{\text{macro}}]$  varies from several decades to a few thousand years. The sample size is $N=\lfloor T_{\text{macro}}/\Delta\rfloor$, ranging from 200 to 2500, and the asymptotic analysis is conducted as $N \rightarrow \infty$. 

Remember that we associate with the time series $(Z_k^\Delta)_{k=1,\ldots, N}$ a continuous process $Z_t = (Z_t)_{t \in [0,1]}$ via the embedding \eqref{eq: continuous embedding}. We will need the following notion of a fractional process.
\begin{defi} \label{def: fractional process}
A random process $(Z_t)_{t \in [0,1]}$ is a fractional process with exponent $H \in (0,1)$ if
\begin{equation} \label{eq: scaling property gen}
k_q h^{Hq} \leq \E\big[|Z_{t+h}-Z_t|^q\big] \leq K_q h^{Hq}
\end{equation}
for every $t \in [0,1]$, $q >0$, sufficiently small $h >0$, and some $k_q, K_q >0$.
\end{defi} 
A prototype example is given by fractional Brownian motion (fBm) with Hurst parameter $H \in (0,1)$.
In the paper, we consider a class of random processes $(Z_t)_{t \geq 0}$ of the form
\begin{equation} \label{eq: fractional processes}
Z_t = \xi_0+g(t)+\int_0^t k_H(t-s)h(Z_s)dB_s,
\end{equation}
where $(B_t)_{t \geq 0}$ is a standard Brownian motion, $g(t)$ and $h(x)$ are smooth real-valued functions defined for every non-negative time $t \geq 0$ and real number $x$, and $k_H$ is a positive kernel defined for every $t \geq 0$ that may be singular at the origin. 
The representation \eqref{eq: fractional processes} is reminiscent of the \cite{MandelbrotVanNess1968} representation of fBm when  $k_H(t-s) \approx (t-s)^{H-{1/2}}$.
The presence of the function $h$ provides modelling flexibility that encompasses fractional stochastic differential equations.
We prove in Appendix \ref{app: proof of frac property} the following.
\begin{prop} \label{prop: frac property}
Assume that $g$ is Lipschitz continuous and $h$ is continuous with at most polynomial growth. Suppose that for some $H \in (0,1)$,  we have
$$c_H(t-s)^{H-1/2} \leq k_H(t-s) \leq C_H(t-s)^{H-{1/2}}$$
for some $0 < c_H, C_H$ and every $0 \leq s \leq t$. Then, any process $(Z_t)_{t \in [0,1]}$ of the form \eqref{eq: fractional processes} such that  $\E[|\xi_0|^q]<\infty$ for every $q >0$  and such that $\E[h(Z_t)] \neq 0$ is a fractional process with exponent $H$ in the sense of Definition \ref{def: fractional process}, with the restriction $q \geq 2$ for the lower bound in Definition \ref{def: fractional process}.
\end{prop}

\subsection{Evidence of roughness: statistical methodology} \label{sec: stat evidence roughness}

We observe
$$Z_{1}^\Delta, Z_{2}^{\Delta},\ldots, Z_N^{\Delta},$$ 
with $N = \lfloor T_{\text{macro}}/\Delta\rfloor$. Equivalently, via the continuous embedding \eqref{eq: continuous embedding}, we discretely observe the continuous process $(Z_t)_{t \in [0,1]}$ at $N$ equidistant times:
$$Z_{\Delta/T_{\text{macro}}}, Z_{2\Delta/T_{\text{macro}}},\ldots, Z_{N\Delta/T_{\text{macro}}}.$$ 
For $q > 0$, define
\begin{equation} \label{eq: empirical scaling moment}
m_N(q, Z) = N^{-1} \sum_{k=1}^N|Z_{k\Delta/T_{\text{macro}}}-Z_{(k-1)\Delta/T_{\text{macro}}}|^q.
\end{equation}
In the same spirit as \cite{Gatheral2022}, our main assumption is a convergence
\begin{equation} \label{eq: scaling}
N^{qs_q}m_N(q, Z) \rightarrow b_q,\;\;\text{as}\;\;N\rightarrow \infty,
\end{equation}
for some $s_q>0$ and $b_q> 0$.
As for the process $(Z_t)_{t \in [0,1]}$, from an asymptotic point of view, it is noteworthy that having in mind a fixed macroscopic timescale $[0,T_{\text{macro}}]$, we equivalently have $\Delta/T_{\text{macro}} \rightarrow 0$. This does not mean that $\Delta$ is small  but rather that the sample size $N = \lfloor T_{\text{macro}}/\Delta \rfloor$ is large, and this is how we conduct statistical inference.
Under the continuity assumption for the random function $t \mapsto Z_t$, Equation \eqref{eq: scaling} is equivalent to say that aggregated rainfall has saturated smoothness $s_q>0$, in the sense that its random paths are in the Besov space $\mathcal B^{s_q}_{q, \infty}([0,1])$ and not in $\mathcal B^{s'}_{q, \infty}([0,1])$ for every $s' > s_q$, see \cite{rosenbaum2009first}. In particular, if $(Z_t)_{t \geq 0}$ is a smooth transformation of an fBm, we have \eqref{eq: scaling} in probability and $s_q = H$ for ever $q$. Finally, the quantity $m_N(q, Z)$ can be seen as an empirical counterpart of
$$\E\big[|Z_{\Delta/T_{\text{macro}}}-Z_0|^q\big] \approx b_q (\Delta/T_{\text{macro}})^{qs_q},$$
provided a law of large numbers holds. Now we see that if $(Z_t)_{t \geq 0}$ is an fBm or, more generally, a fractional process of the form \eqref{eq: fractional processes} as in Section \ref{sec: model large scale data} above, we can recover the Hurst exponent $H$ by means of empirical moments of increments \eqref{eq: empirical scaling moment} at several scales and for several values $q$. The mathematical link is granted by the convergence \eqref{eq: scaling} that asserts the identity $s_q=H$ for every $q > 0$.

\subsection{Statistical evidence of roughness at large time scales} \label{sec: evidence roughness}

We analyse two categories of data:
\begin{itemize}
\item[$\bullet$] Weather station data: observational data spanning more than 200 consecutive years, including the data from the Global Historical Climatology Network (GHCN) used in \cite{markonis2016}\footnote{Available at \url{https://www.ncei.noaa.gov/data/ghcnm/v4/precipitation/archive/}. We did not find the data for Paris and Manchester used in \cite{markonis2016} in this database.}, as well as the dataset from Collegio Romano (Rome, Italy) described in \cite{volpi2024}. These data are among the longest available precipitation records and are all located in Europe; see 
Table~\ref{tab:H_weather_station} for their coordinates. Four of the GHCN stations (Lille, Marseille, Strasbourg, and Toulouse) are the same weather stations as the M\'et\'eo France stations used in the microscopic experiments of Section~\ref{sec: stat criticality}.
\item[$\bullet$] Paleoclimatic data from indirect measurements (tree rings, lake sediments, pollen), including some of the data from the meta-study of \cite{iliopoulou2018}\footnote{Available at \url{https://www.ncei.noaa.gov/access/paleo-search/}.}. 
\end{itemize}



For a time series $X$ of length $N$ and aggregation time step $\Delta$, we estimate 
$m(q,h) = \mathbb{E}[|X_{t+h} - X_t|^q] \approx C h^{\xi_q}$, with $\xi_q = Hq$, 
for $h = k\Delta/T_{\text{macro}}$, $k = 1, \ldots, \min(60, N/10)$. We regress $\log(m(q,h))$ against $\log(h)$ and obtain an estimate of $\xi_q$ as the slope of this linear regression. We then regress $\xi_q$ against $q$ and estimate $H$ as the slope of this second empirical relationship.

\medskip
Concerning the weather station data, we work with different datasets spanning more than $N=200$ consecutive years, with $\Delta = 1$ year. The estimates of $H$, as well as information about the weather stations, are provided in Table~\ref{tab:H_weather_station}. An illustration for the Collegio Romano station~\cite{volpi2024} is shown in Figure~\ref{fig:romano}.

\medskip
Concerning the paleoclimatic data, we consider four datasets based on different reconstruction sources:
\begin{itemize}
    \item Tree rings, which require preprocessing in order to recover low-frequency climate variations ``are embedded in the data along with other long-term effects such as tree ageing and population dynamics''. Three main methods are employed. The first method, corresponding to the dataset from~\cite{yang2014}, consists of detrending combined with pre-whitening (DP), in which a parametric trend is removed and ARMA-type pre-whitening is applied to reduce the autocorrelation induced by climate change. The second method, corresponding to the dataset from \cite{buntgen2011}, uses regional curve standardisation (RCS), which consists of averaging tree-ring measurements from the same region and species according to their biological age in order to estimate a common regional growth curve and remove age-related growth trends~\cite{helama2017}. The third method relies on neural-network regression (NN) and corresponds to the dataset from~\cite{ni2002} (Arizona Climate Dataset 1).
    \item Lake sediments, corresponding to the dataset from~\cite{romero2011}.
    \item Pollen, corresponding to the dataset from~\cite{viau2009}.
\end{itemize}

The different values of $H$ are reported in Table~\ref{tab:H_paleo}. The tree-ring data reconstructed using RCS, together with the different linear regressions involved in the estimation of $H$, are displayed in Figure~\ref{tab: TreeCE11}.

\medskip
The results consistently support compatibility with a rough model at large time scales, with a Hurst exponent between $10^{-2}$ and $10^{-1}$, subject to some variability. Some values from the weather stations are very close to zero, or even negative, possibly due to the limited amount of available data. Paleoclimatic data help mitigate this issue, as they are based on a larger number of observations. In Table~\ref{tab:H_paleo}, it can also be seen that the pollen-based reconstructions exhibit a much higher $H$, which may be related to the coarser 100-year aggregation. In Figure~\ref{tab: TreeCE11_simu}, we provide a simulation of the model $\sigma W_t^H$ using parameters estimated from the Central Europe dataset, where $W^H_t$ is an fBm with the Hurst parameter $H=0.06$ and $\sigma >0$\footnote{$\sigma$ is estimated from the intercept of $\log(m(2,h))$.}. The close agreement between Figures~\ref{tab: TreeCE11} and~\ref{tab: TreeCE11_simu} indicates that the model captures the key statistical features of the observed data, lending support to its validity and plausibility.
%
%
%

\begin{figure}[H]
    \centering
    \begin{tabular}{cc}
        \includegraphics[width=0.35\linewidth]{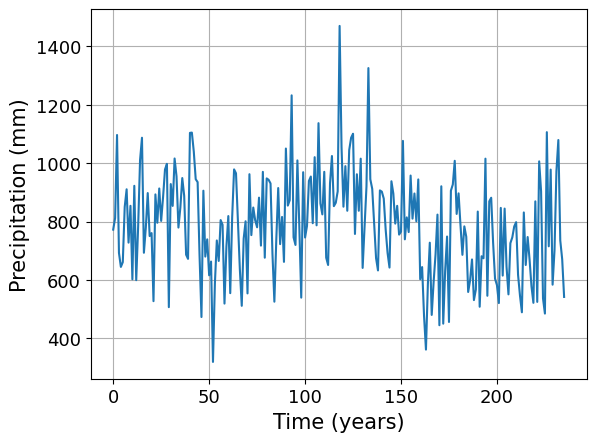}     & \includegraphics[width=0.35\linewidth]{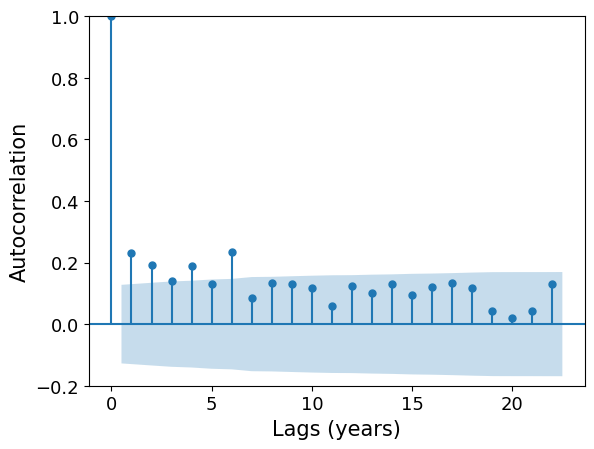}  \\
          \includegraphics[width=0.35\linewidth]{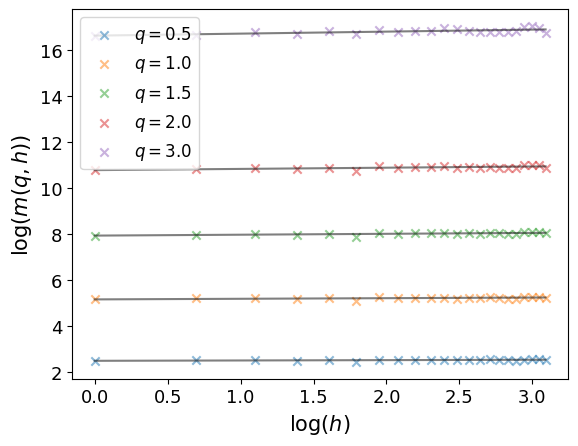}         &      \includegraphics[width=0.35\linewidth]{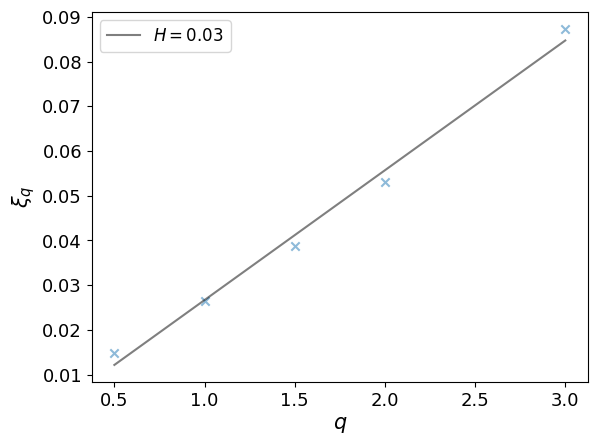} 
    \end{tabular}  
    \caption{\small \it Rainfall in Romano (Italy), $\Delta=1$ year, $N=236$. Upper Left: Data; Upper Right: autocorrelation of the data with Bartlett standard error bands; Lower Left: Scaling law of the structure function; Lower right: estimation of $H$ by linear regression on the slopes of the structure function.\label{fig:romano}} 
    \end{figure}

\begin{table}[H]
    \centering
\begin{tabular}{llllrr}
\toprule
Name & ID & Lon & Lat & $N$ & $H$ \\
\midrule
Edinburgh & UKXLP329564 & 3.2$^\circ$W & 55.9$^\circ$N & 215 & 0.01 \\
Hoofddorp & NLE00100503 & 4.7$^\circ$E & 52.3$^\circ$N & 291 & 0.02 \\
Kew Gardens & UKMLP003775 & 0.3$^\circ$W & 51.5$^\circ$N & 303 & -0.01 \\
\midrule
Klagenfurt & AUM00011231 & 14.3$^\circ$E & 46.6$^\circ$N & 210 & 0.04 \\
Lille & FRE00104040 & 3.1$^\circ$E & 50.6$^\circ$N & 242 & 0.01 \\
Lund & SWE00137568 & 13.2$^\circ$E & 55.7$^\circ$N & 278 & 0.00 \\
\midrule
Marseille & FR000007650 & 5.2$^\circ$E & 43.4$^\circ$N & 277 & -0.00 \\
Milano & ITMLP016080 & 9.3$^\circ$E & 45.5$^\circ$N & 236 & 0.02 \\
Oxford & UK000056225 & 1.3$^\circ$W & 51.8$^\circ$N & 259 & 0.00 \\
\midrule
Padua & ITXLP330782 & 12$^\circ$E & 45.4$^\circ$N & 250 & 0.03 \\
Podehole & UKXLP329602 & 0.1$^\circ$W & 52.8$^\circ$N & 269 & 0.01 \\
Praha & EZE00100082 & 14.4$^\circ$E & 50.1$^\circ$N & 201 & -0.01 \\
\midrule
Rome &  & 12.47$^\circ$E & 41.9$^\circ$N & 236 & 0.03 \\
Strasbourg & FR000007190 & 7.6$^\circ$E & 48.5$^\circ$N & 224 & 0.02 \\
Toulouse & FR000007630 & 1.4$^\circ$E & 43.6$^\circ$N & 217 & 0.02 \\
\midrule
Uppsala & SWE00139148 & 17.6$^\circ$E & 59.9$^\circ$N & 252 & 0.01 \\
\bottomrule
\end{tabular}
    \caption{\small \it Estimates of $H$ for the different weather station data, including 15 dataset from the Global Historical Climatology Network with identification code reported in ID, and the dataset of Rome from~\cite{volpi2024}. $\Delta =1$ year and $N$ is the number of consecutive years in the data.}
    \label{tab:H_weather_station}
\end{table}

\begin{figure}[H]
    \centering
    \begin{tabular}{cc}
        \includegraphics[width=0.39\linewidth]{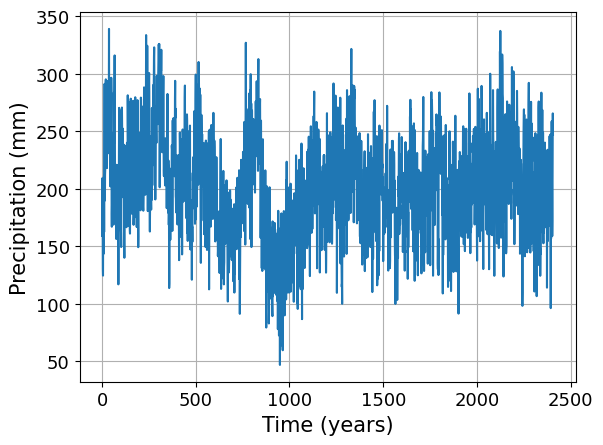}     & \includegraphics[width=0.39\linewidth]{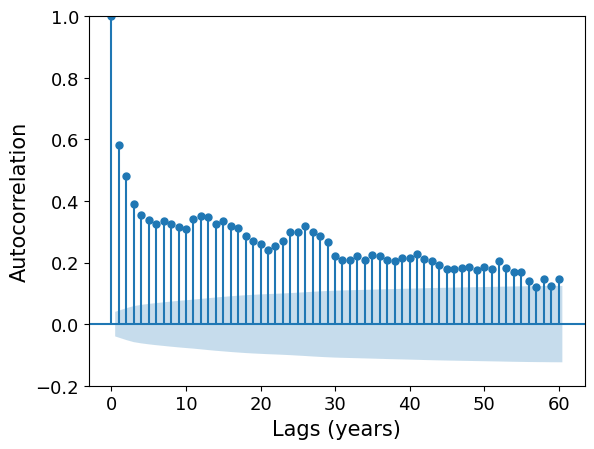}  \\
          \includegraphics[width=0.39\linewidth]{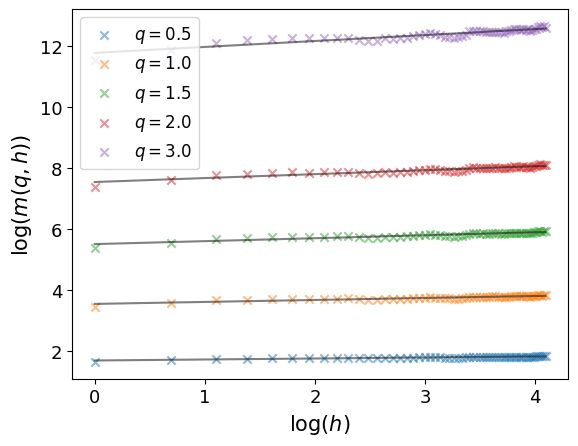}         &      \includegraphics[width=0.39\linewidth]{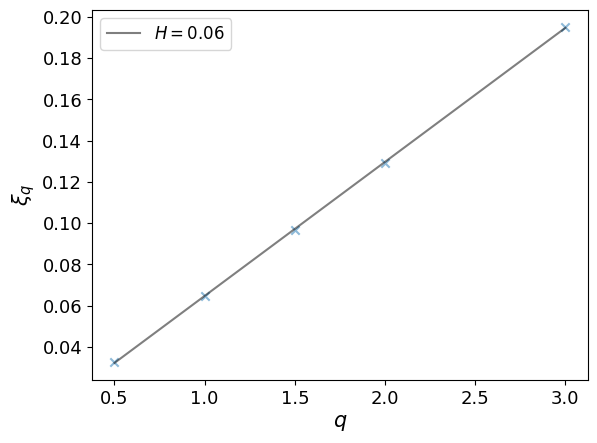}   
            \end{tabular}   
      \caption{\label{tab: TreeCE11}\small \it Tree rings measurements: Central Europe data from \cite{buntgen2011} (total precipitation over the months of April, May, and June), $\Delta=1$ year, $N=2407$. Upper Left: Data, pre-processed by RCS reconstruction; Upper Right: autocorrelation of the data with Bartlett standard error bands; Lower Left: Scaling law of the structure function; Lower right: estimation of $H$ by linear regression on the slopes of the structure function.} 
          \end{figure}
          
\begin{figure}[H]
    \centering
    \begin{tabular}{cc}
        \includegraphics[width=0.39\linewidth]{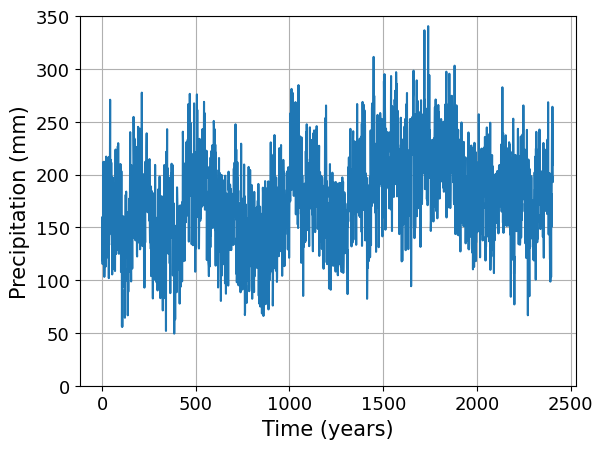}     & \includegraphics[width=0.39\linewidth]{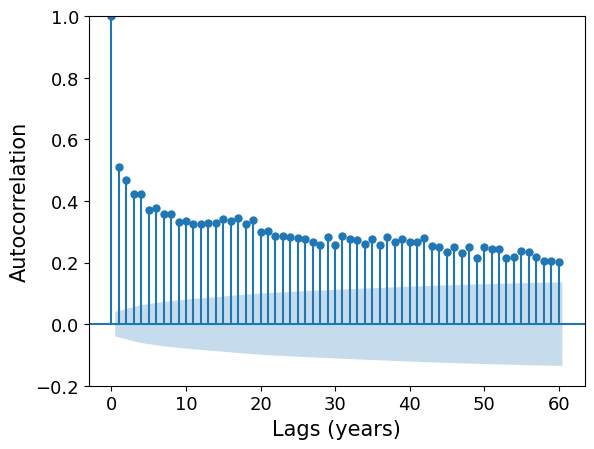}  \\
          \includegraphics[width=0.39\linewidth]{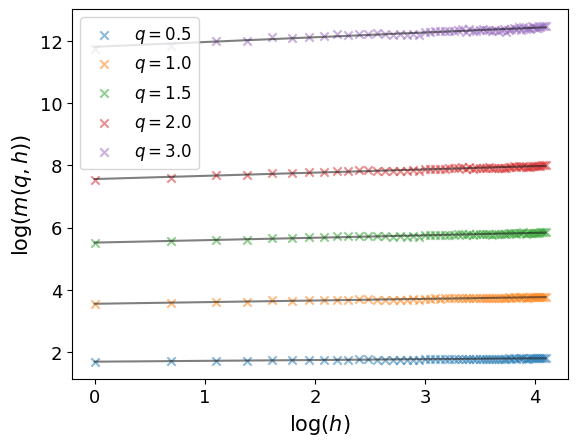}         &      \includegraphics[width=0.39\linewidth]{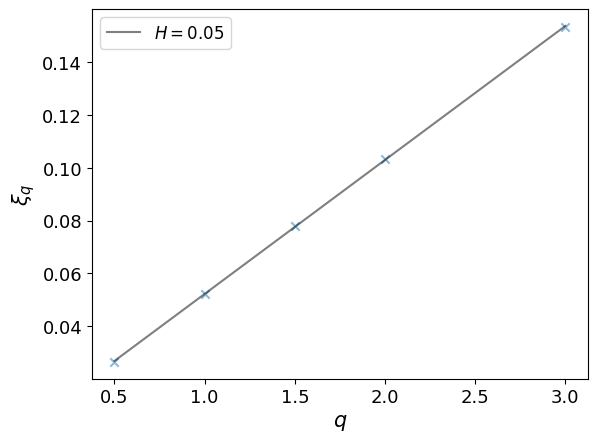}   
            \end{tabular}   
      \caption{\label{tab: TreeCE11_simu}\small \it  Simulation of $\sigma W^H_t$ where $W^H$ is an fBm with parameter $H$ and $\sigma$ and $H$ are estimated from the Central Europe data from \cite{buntgen2011}. Upper Left: Data, pre-processed by RCS reconstruction; Upper Right: autocorrelation of the data with Bartlett standard error bands; Lower Left: Scaling law of the structure function; Lower right: estimation of $H$ by linear regression on the slopes of the structure function.} 
          \end{figure}
          
    \begin{table}[H]
    \centering
\begin{tabular}{llllrrr}
\toprule
Type & Region & Lon & Lat & $N$ & $\Delta$ & $H$ \\
\midrule
Tree rings (RCS) & Central Europe & 6-20$^\circ$E & 45-53$^\circ$N & 2407 & 1 & 0.06 \\
Tree rings (DP) & Tibet & 97-100$^\circ$E & 37-39$^\circ$N & 3512 & 1 & 0.04 \\
Tree rings (NN) & Arizona, USA & 115-113$^\circ$W & 34-37$^\circ$N & 989 & 1 & 0.03 \\
\midrule
Lake sediments & La Cruz, Spain & 2$^\circ$W & 40$^\circ$N & 372 & 1 & 0.07 \\
Pollen & Central Boreal, Canada & 120-80$^\circ$W & 50-70$^\circ$N & 119 & 100 & 0.26 \\
\bottomrule
\end{tabular}
    \caption{Estimates of $H$ for the different paleoclimatic data. $\Delta$ is the aggregation timestep in years and $N$ is the number of consecutive data.}
    \label{tab:H_paleo}
\end{table}

\section{Discussion} \label{sec: discussion}

%
%
%
%
%
%
%

\subsection{Hurst and Mandelbrot revisited}

In his seminal works on reservoir design~\cite{hurst1951, hurst1956a, hurst1956b}, Hurst investigated the long-term behaviour of several annual rainfall series, focusing in particular on the so-called \emph{rescaled range} statistic $R(t,h)/S(t,h)$. This statistic is defined by
\[
R(t,h) ={} \max_{k=0,\ldots, h}
\Big(
\sum_{i=t+1}^{t+k} Z_i^{\Delta}
- \frac{k}{h} \sum_{i=t+1}^{t+h} Z_i^{\Delta}
\Big) \\
-
\min_{k=0,\ldots, h}
\Big(
\sum_{i=t+1}^{t+k} Z_i^{\Delta}
- \frac{k}{h} \sum_{i=t+1}^{t+h} Z_i^{\Delta}
\Big),
\]
and measures the deviation of cumulative rainfall from its mean trend over the considered time window. The quantity
\[
S^2(t,h)
=
h^{-1} \sum_{i=t+1}^{t+h} \left(Z_i^{\Delta}\right)^2
-
\Big(
h^{-1} \sum_{i=t+1}^{t+h} Z_i^{\Delta}
\Big)^2
\]
denotes the empirical variance of the data between times $t+1$ and $t+h$. Hurst observed a behaviour that was unexpected at the time:
\[
\frac{R(t,h)}{S(t,h)} \sim h^{H_{\text{Hurst}}},
\qquad H_{\text{Hurst}} > \frac12.
\]
This finding apparently contradicts the hypothesis of independent observations, which would correspond to the classical case $H_{\text{Hurst}} = 1/2$. Hurst showed that this phenomenon occurs across several different data sets, suggesting a form of universality.\\

Mandelbrot~\cite{mandelbrot1968} subsequently proposed modelling the cumulative rainfall over a time interval $[0,t]$ at macroscopic time scales, $\int_0^t Z_s \, ds$,
by an fBm $(W_t^{H_{\text{Hurst}}})_{t \geq 0}$ with parameter $H_{\text{Hurst}} > 1/2$, in order to account for the Hurst effect observed in the rescaled range statistic. Since the formal derivative of $W^{H_{\text{Hurst}}}_t$ corresponds to a fractional Gaussian noise (fGn) and is too irregular to be studied directly as an ordinary stochastic process\footnote{``Unfortunately, the derivative $B_H'(t)$, called
'fractional Gaussian noise,' is too irregular to
be studied directly. As we interpolated the integral of the independent Gauss process by Brownian motion, we must now replace $X(t)$ by $B_{H}(t+1)-B_{H}(t)$''.}, the authors are led to consider a discrete-time model instead, by setting
\[
Z_k^{\Delta} = Z_{k\Delta/T_{\text{macro}}} = W^{H_{\text{Hurst}}}(k\Delta/T_{\text{macro}}) - W^{H_{\text{Hurst}}}((k-1)\Delta/T_{\text{macro}}).
\]

Our approach is different and is closer in spirit to the work of~\cite{Gatheral2022} on the volatility of financial time series: we model the process $Z_t$ directly, rather than its integral, using a rough fractional process, for instance, an fBm, and $Z_k^{\Delta} = Z_{k\Delta/T_{\text{macro}}}$ is a discrete time sampling of $Z_t$. The resulting regularity is no longer characterised by $H_{\text{Hurst}} > 1/2$, but rather by $H < 1/2$, consistently with the fact that the object under consideration is less regular than the one studied by~\cite{mandelbrot1968}, namely its integral.   \\

In Figure~\ref{fig:charleston}, we illustrate the results of the experiment described in Section~\ref{sec: evidence roughness}, applied to the Charleston rainfall series analysed in~\cite{mandelbrot1969}. While the authors of~\cite{mandelbrot1969} report $H_{\text{Hurst}} = 0.89$ for the integrated process, our analysis consistently yields $H = 0.11$ for the underlying process $Z_t$, with no contradiction. Note that ~\cite{georgakakos1994} also studied the scaling behaviour of $Z_t$, rather than $\int_0^t Z_s\,ds$, for storms recorded in Iowa City, reporting values of $H$ below $1/2$.\\

\begin{figure}[H]
    \centering
    \begin{tabular}{cc}
        \includegraphics[width=0.35\linewidth]{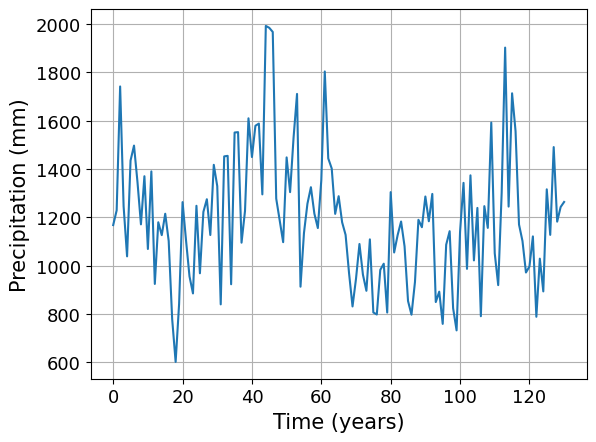}     & \includegraphics[width=0.35\linewidth]{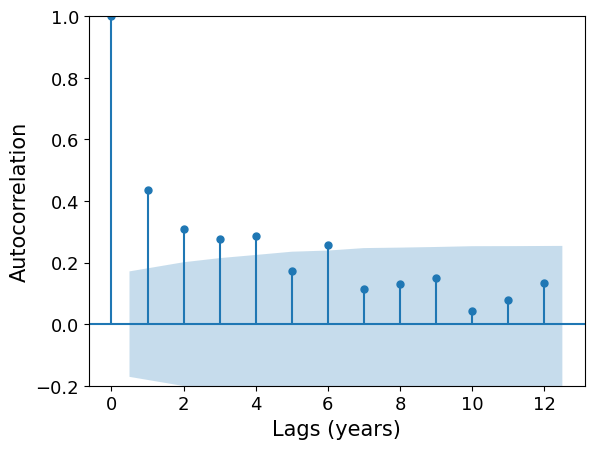}  \\
          \includegraphics[width=0.35\linewidth]{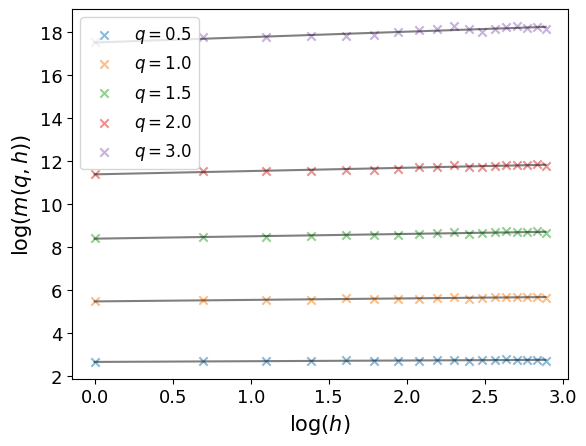}         &      \includegraphics[width=0.35\linewidth]{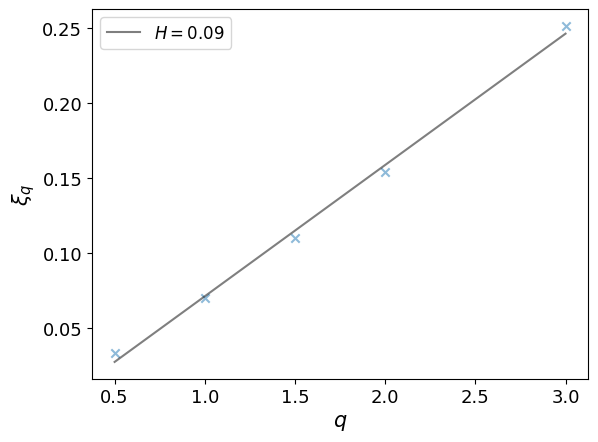} 
    \end{tabular}  
    \caption{\small \it Rainfall in Charleston (USA) between 1832 and 1962, $\Delta=1$ year, $N=131$. Upper Left: Data; Upper Right: autocorrelation of the data with Bartlett standard error bands; Lower Left: Scaling law of the structure function; Lower right: estimation of $H$ by linear regression on the slopes of the structure function.\label{fig:charleston}} 
    \end{figure}

Our framework departs from the seminal approach of~\cite{mandelbrot1968} in two important ways:
\begin{itemize}
    \item The process $(Z_t)$ is stationary by construction in~\cite{mandelbrot1968}. 
    The stationarity for fractional processes such as those in our macroscopic limits is a delicate issue, see \cite{gnabeyeu2025inhomogeneous,gnabeyeu2026}. However, the small values for our Hurst parameter $H$ ensure the consistency of the approach over a very large range of time scales. This is because the scaling between a time interval $[0,1]$ and $[0,T]$ is in $T^H$, which is very slowly increasing with $T$ (with large $H$ some kind of mean reversion force would be needed to avoid exploding behaviours). Consequently, with $H$ small, the laws of the dynamics are very slowly distorted with time. This is exactly what we observe on data, where after a simple time scaling (and no space scaling), 300 years of data look very similar to 3000 years of data (up to the number of data points in each graphs), see Figure~\ref{fig:CE11multiscale}. Such phenomenon was already observed for the volatility of financial asset in \cite{Gatheral2022}.
    \item The fGn framework describes  long-memory processes, characterised by an autocorrelation function $\rho(\tau)$ that decays slowly as $\tau^{2H_{\text{Hurst}}-2}$ when $\tau \to \infty$, with $1 > H_{\text{Hurst}} > 1/2$, as observed for example in Figure~\ref{tab: TreeCE11}. For a stationary process, the autocorrelation function is directly related to the aggregated variance
    \begin{equation}
    \label{eq:Ch}
    C(h)=\mathrm{Var}\Big(\sum_{i=t}^{t+h} Z_i^\Delta\Big).
    \end{equation}
    In particular, if
    \[
    \rho(\tau)\sim \tau^{2H_{\text{Hurst}}-2}
    \qquad \text{as}\; \tau\to\infty,
    \]
    then
    \begin{equation}
    \label{eq:varFGN}
    C(h)\sim h^{2H_{\text{Hurst}}}
    \qquad \text{as}\; h\to\infty.
    \end{equation}
    Relation~\eqref{eq:varFGN} therefore gives a way to estimate $H_{\text{Hurst}}$ by regressing $\log(C(h))$ against $\log(h)$ for large values of $h$. This approach is used, for instance, in~\cite{marani2003, markonis2016, iliopoulou2018}, where the authors obtain estimates $H_{\text{Hurst}} > 1/2$ for several precipitation records. Since our model is not stationary, the usual notions of autocorrelation and long memory are no longer well defined. Nevertheless, we show that our model reproduces the same empirical autocorrelation structure as the observed data; see Figure~\ref{tab: TreeCE11_simu}, which should be compared with Figure~\ref{tab: TreeCE11}. In Figure~\ref{fig:simu_var_emp}, we compare the linear regression of the empirical $\log(C(h))$ against $\log(h)$, whose slope is equal to $2H_{\text{Hurst}}$ in the fGn framework, for both the Central European paleoclimate data and simulations of the model $\sigma W_t^H$ displayed in Figure~\ref{tab: TreeCE11_simu}, with $H=0.06$. In both cases, we obtain a very similar estimate, namely $H_{\text{Hurst}} \approx 0.87$. This phenomenon has already been identified and investigated in~\cite{Gatheral2022}, where the authors show that an fBm can generate such spurious long memory; see in particular Section~4.
\end{itemize}

\begin{figure}[H]
    \centering
    \includegraphics[width=0.8\textwidth]{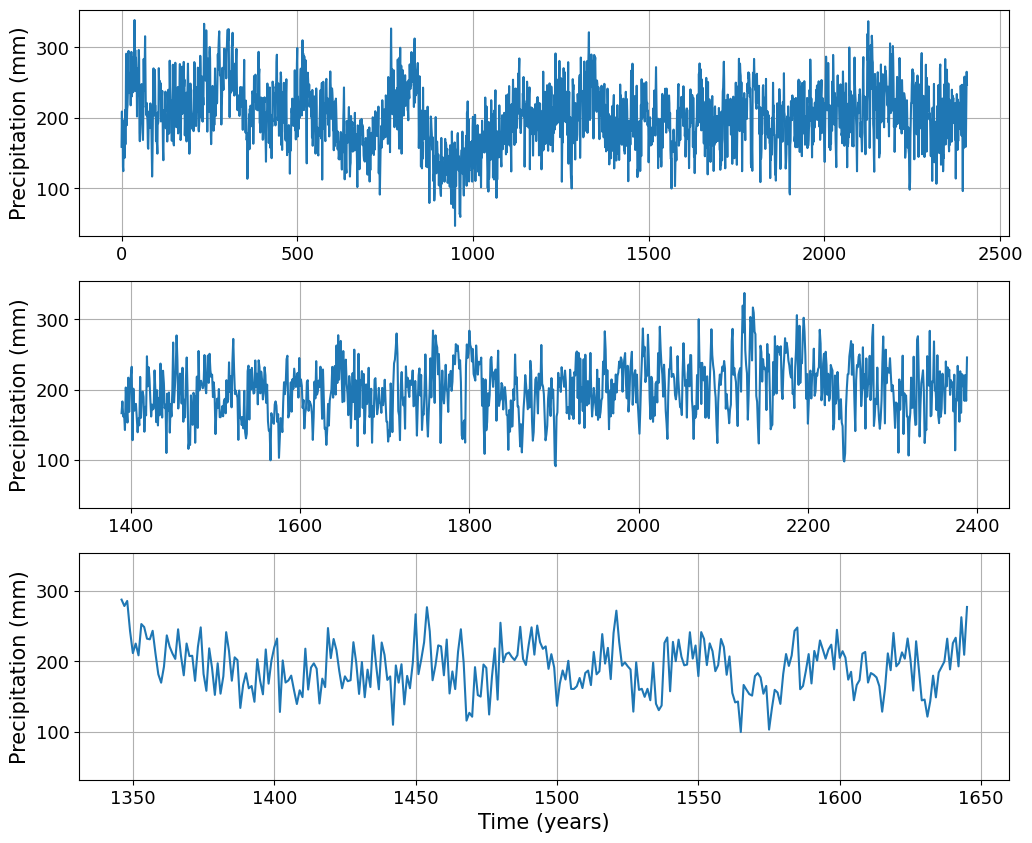}
    \caption{\label{fig:CE11multi}\small Central Europe data from \cite{buntgen2011} over 2407 years (top), 1000 years (middle), and 300 years (bottom).}
    \label{fig:CE11multiscale}
\end{figure}

\begin{figure}[H]
    \centering
    \begin{tabular}{cc}
        \includegraphics[width=0.39\linewidth]{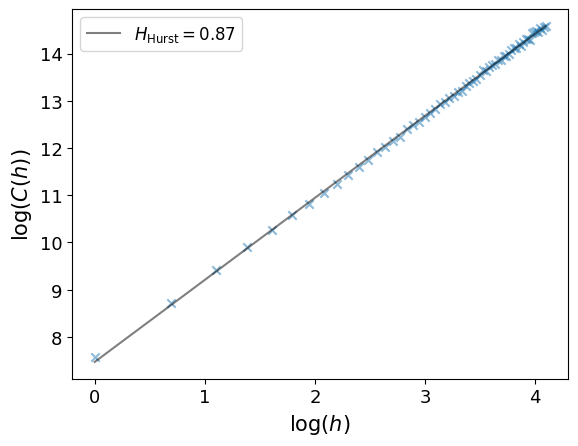}     & \includegraphics[width=0.39\linewidth]{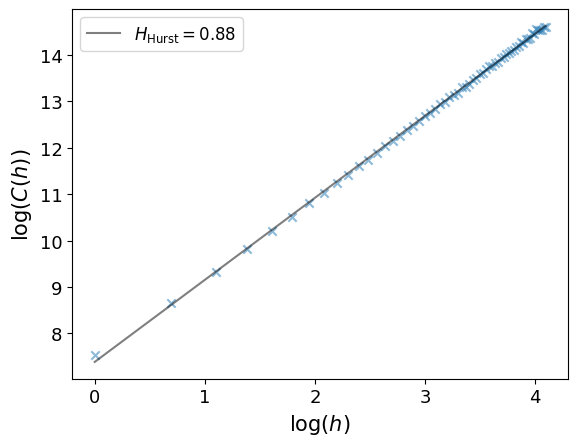} 
            \end{tabular}   
      \caption{\label{fig:simu_var_emp}\small \it Regression of the empirical version of $\log(C(h))$ against $\log(h)$ where $C(h)$ is defined in \eqref{eq:Ch} for the Central Europe data from \cite{buntgen2011} and for a simulation of $\sigma W^H_t$ where $W^H$ is an fBm with parameter $H$ and $\sigma$ and $H$ are estimated from this dataset. The quantity $H_{\text{Hurst}}$ is equal to the estimated slope of the linear regression divided by 2.} 
          \end{figure}


\subsection{From roughness to multifractality}

 In geophysics, there is a wide consensus that rainfall dynamics involve a wide range of time scales, and that scaling laws are somewhat ubiquitous. In particular, power-laws have been observed and characterised, at least some range of scales. In this context, multifractal analysis, first proposed by \cite{SchertzerLovejoy1987} has been applied with some success, despite several experimental biases. For a review, see \cite{LovejoySchertzer2007, LovejoySchertzer2010, VerrierMalletBarthes2011}.\\
 
 Close to our study (and with our notation), we consider the following statistics extracted from the studies of Lovejoy and Schertzer:
 \begin{itemize}
     \item The moment trace function 
     $$m_k(q,\Delta) =\E\big[|Z_k^\Delta|^q\big] = \E\big[\big|\int_{(k-1)\Delta T}^{k \Delta T} Y_s ds\big|^q\big].$$
     \item The fluctuation exponent $$f_{k\Delta}(h) = \E\big[\big|Z_{k+h\Delta^{-1}}^\Delta-Z_k^{\Delta}\big|\big] = \E[\big|Z_{k\Delta+h}-Z_{k\Delta}\big|].$$
 \end{itemize}
The universal multifractal (UM) law of Schertzer-Lovejoy predicts that 
$$m_k(q,\Delta) \approx C_q\Delta^{\phi(q)},$$
when $\Delta$ is small, where 
$$\phi(q) = q(1-H_{SL})-K(q).$$
In the universal log-stable case, 
$$K(q) = \frac{C_1}{\alpha-1}(q^\alpha-q),\;\;C_1\geq0, \;\;\alpha \in [0,2].$$
Here $q \mapsto K(q)$ is the (convex) moment scaling function, characterizing intermittency, with $K(0)=K(1)=0$. The conservation parameter $H_{SL}$ measures the deviation from strict conservation, in which case $H_{SL}=0$. In particular, for $\alpha=2$, we find back the log-normal multifractal model. Note that a stationary assumption on $(Y_t)$ imposes that $m_k(q,\Delta)$ does not depend on $k$. 

At first glance, the fact that $K(q) \neq Hq$ seems in contradiction with the approximation of a fractional process. However there is no contradiction at all as the moment trace function does not measure the increments of $Z_k^\Delta$. Closer to our approach is the fluctuation exponent $f_{k\Delta}(h)$ which corresponds to our Definition \ref{def: fractional process} for $q=1$. Here, we shall therefore anticipate from our empirical study a result of the form
$$f_{k\Delta}(h) \approx \Delta\, h^{H}$$
for some rough parameter $H \in (0,1/2)$. The factor $\Delta$ in front comes from the definition of the aggregated rainfall, it is simply a dimensional prefactor. The UM theory of Schertzer and Lovejoy predicts the relationship 
$$f_{k\Delta}(h) \approx \Delta\,h^{H_{SL}},\;\;\Delta \ll h.$$
The dominant finding across the Lovejoy-Schertzer literature is that rainfall is close to a conservative process, meaning $H_{SL} \approx 0$
in most observational studies\footnote{This is consistent with the physical intuition that total rainfall accumulation is approximately conserved across scales.}. This result is in accordance with our independent approach where we consistently find $H \approx 0$.\\

Note that a parallel can be made again between rainfall and volatility studies regarding multifractal properties. More precisely, it is shown in \cite{Neuman2018fBmZero} that when $H$ is small, processes based on integrals of fBM can reproduce similar multifractal properties\footnote{It is proved in \cite{Neuman2018fBmZero} that once properly normalized, the fBM converges as $H\to0$ to a random Gaussian distribution that is close to a log-correlated Gaussian field thereby establishing the rigorous bridge between the rough volatility class and multifractal processes.} as for example the model of \cite{Bacry2001MRW, Bacry2003LogInf} for describing aggregated volatility. This model corresponds to the log-normal case $\alpha=2$ of the UM approach with
$$\phi(q) = q-\lambda^2q(q-1),$$
where the intermittency parameter $\lambda$ is empirically found to be of order $10^{-2}$, smaller than the $C_1$ parameter in rainfall models. Also, the conservation parameter is systematically set to $0$. See also \cite{Wu2022LogSfBM} for connections between rough and multifractal models.
%



\section*{Acknowledgements}

The authors acknowledge support from the FiME Lab (Institut Europlace de Finance) and from the ILB Chair {\it Artificial Intelligence and Quantitative Methods for Finance} at University Paris Dauphine-PSL.

\bibliographystyle{apalike}
        \bibliography{biblio}

\section{Appendix}  \label{sec: appendix}

\subsection{The BL and NS models as latent degenerate Hawkes processes} \label{app: proof prop embedding}

\begin{prop} \label{prop: embedding} The BL rain cell arrival process $(N_t)_{t \geq 0}$ can be written as $N_t = N_t^{1} + N_t^{2}$, extracted from a three-dimensional Hawkes process $\boldsymbol{N} = (N_t^{1}, N_t^{2}, N_t^{3})_{t \geq 0}$ 
having intensity $\boldsymbol{\lambda} = (\lambda^1_t, \lambda^2_t, \lambda_t^3)_{t \geq 0}$ given by
\[
\boldsymbol{\lambda}_t = \begin{pmatrix}
\mu_{BL} \\ 0 \\ 0\end{pmatrix} + \int_{[0,t)} \begin{pmatrix}
0 & 0 & 0 \\ \nu_{BL} & 0 & - \nu_{BL} \\ \gamma_{BL} & 0 & -\gamma_{BL}
\end{pmatrix}d\boldsymbol{N}_s.\]
This is a degenerate linear Hawkes process (in the sense of allowing for negative entries in its kernel matrix). It has a well-defined nonlinear representation
in the sense of Definition 1 in \cite{DelattreFournierHoffmann2016}: we also have $\boldsymbol \lambda = \boldsymbol h \circ \boldsymbol \lambda$, with
$\boldsymbol h(x_1, x_2, x_3) = ((x_1)_+, (x_2)_+, (x_3)_+)$ with $x_+ = \max(x, 0)$ so that $\boldsymbol h$ is nonnegative and Lipschitz continuous.\\ 

For the NS model, if $C$ is Poisson distributed with mean $\nu_{NS}$, the NS  rain cell arrival process can be written as the second component $(N_t^{2})_{t \geq 0}$ of a bivariate linear Hawkes process $\boldsymbol{N} = (N_t^{1}, N^{2}_{t})_{t \geq 0}$ having intensity $\boldsymbol{\lambda} = (\lambda^1_t, \lambda^2_t)_{t \geq 0}$ given by
\[
\boldsymbol{\lambda}_t =
\begin{pmatrix}
\mu_{NS} \\ 0 \end{pmatrix} + \int_{[0,t)} \begin{pmatrix}
0 & 0 \\ \nu_{NS} f_{NS}(t-s) & 0 \end{pmatrix}d\boldsymbol{N}_{s}.
\]  
\end{prop}
The second part of Proposition \ref{prop: embedding} is Proposition 1 in \cite{smith1985} for the special case where $f_{NS}$ is exponential. 

\begin{proof}
We first consider the Bartlett-Lewis model. Let $N^1_t = \sum_{i \geq 1} {\bf 1}_{\{\tau_i \leq t\}}$ be a Poisson process with intensity $\lambda^1_t = \mu_{BL}$ and arrival times $(\tau_i)_{i \geq 1}$ representing the arrival of parent rain cells. Let $(\varepsilon_i)_{i \geq 1}$ be a sequence of independent exponential random variables with common parameter $\gamma_{BL}$, independent of $(N^1_t)_{t \geq 0}$, representing the lifetime of each parent rain cell. Define, for $t \geq 0$ and $i \geq 1$,
$$N_i^3(t) = {\bf 1}_{\{\tau_i+\varepsilon_i \leq t\}},\;\;\lambda_i^3(t) = \gamma_{BL} {\bf 1}_{\{\tau_i < t \leq \tau_i+\varepsilon_i\}}.$$
The $N_i^3$ are degenerate counting processes (in the sense of counting only one event\footnote{namely the time of death of the parent rain cell $i$.}) with stochastic intensity 
$\lambda_i^3$. During the lifetime of the parent rain cell $i$, conditional on $N^1$ and $(\varepsilon_i)_{i \geq 1}$,  we draw a Poisson process with intensity $\nu_{BL}$. This results in a family of independent (conditional on $N^1$ and $(\varepsilon_i)_{i \geq 1}$) counting processes $(N^2_i)_{i \geq 1}$ with stochastic intensities
$$\lambda_i^2(t) = \nu_{BL}{\bf 1}_{\{\tau_i < t \leq \tau_i+\varepsilon_i\}}.$$
The counting process $N^2_t = \sum_{i=1}^{N_t^1}N_i^2(t)$
has intensity 
$$\lambda_t^2 = \nu_{BL}\sum_{i = 1}^{N_t^1}{\bf 1}_{\{\tau_i < t \leq \tau_i+\varepsilon_i\}} = \nu_{BL} (N_{t^-}^1-N_{t^-}^3),$$
where
$$N_t^3 = \sum_{i = 1}^{N_t^1}N_{i}^3(t)$$
is a counting process with intensity
$$\lambda^3(t) = \gamma_{BL}\sum_{i \geq 1}{\bf 1}_{\{\tau_i < t \leq \tau_i+\varepsilon_i\}} = \gamma_{BL} (N_{t^-}^1-N_{t^-}^3).$$
Finally, the total number of parent and children rain cells at time $t$ is given by $N_t^1+N_t^2$ and the result follows, noting that the components of the point process $\boldsymbol{N} = (N^1, N^2, N^3)$ never jump simultaneously so that its law is entirely characterised by $\boldsymbol{\lambda} = (\lambda^1, \lambda^2,\lambda^3)$. By construction we always have $N_{t^-}^1 - N_{t^-}^3 = (N_{t^-}^1 - N_{t^-}^3)_+$, therefore
\[
\boldsymbol{\lambda}_t = \left( \begin{pmatrix}
\mu_{BL} \\ 0 \\ 0\end{pmatrix} + \int_{[0,t)} \begin{pmatrix}
0 & 0 & 0 \\ \nu_{BL} & 0 & - \nu_{BL} \\ \gamma_{BL} & 0 & -\gamma_{BL}
\end{pmatrix}d\boldsymbol{N}_s
\right)_+
\]
entrywise, where $x_+ = \max(x, 0)$, hence $\boldsymbol N$ is a well-defined nonlinear Hawkes process in the sense of Definition 1 in \cite{DelattreFournierHoffmann2016}.\\


The proof for the Neyman-Scott model is similar: if the number of children rain cells in the NS model is Poisson distributed with parameter $\nu_{NS}$, and if the distance between the parent and the children rain cells has probability distribution function $f_{NS}$, each parent celle $i$ produces cells according to an inhomogeneous Poisson process with intensity $\lambda_i(t) = \nu_{NS} f_{NS}(t-\tau_{i}){\bf 1}_{\{\tau_i < t\}}$, where the $(\tau_{i})_{i \geq 1}$ are the arrival times of parent rain cells following a Poisson process $N^{1}$ with intensity $\mu_{NS}$. Since each parent cell generates rain cells independently, the total number of rain cells is a counting process $N^2$ with intensity 
$$\sum_{i=1}^{N^{1}_t}    \nu_{NS} f_{NS}(t-\tau_{i}){\bf 1}_{\{\tau_i < t\}} = \int_{[0,t)} \nu_{NS} f_{NS}(t-s)dN_{s}^1.$$ 
The process $\boldsymbol{N} = (N^1,N^2)$ has intensity  
\[
\boldsymbol{\lambda}_t =
\begin{pmatrix}
\mu_{NS} \\ 0 \end{pmatrix} + \int_{[0,t)} \begin{pmatrix}
0 & 0 \\ \nu_{NS} f_{NS}(t-s) & 0 \end{pmatrix}d\boldsymbol{N}_{s}
\]  
and the Neyman-Scott process is simply\footnote{Parent rain cells are not included in $N_t$ in the original Neyman-Scott model.} 
$N^2$. The proof of Proposition \ref{prop: embedding} is complete.
\end{proof}

\subsection{Proof of Proposition \ref{prop: expo second-order}} \label{app: proof prop expo second-order}

We only sketch the proof, as it actually follows from the computations of Propositions \ref{prop: fourier spectral} and \ref{prop: multiscale contrast} given below. 
From \eqref{eq: var en fourier} and \eqref{eq: cov en fourier} in the proof of Proposition \ref{prop: multiscale contrast} below, we have
$$\mathrm{Var}(Y_i^\delta)   = 2 \int_0^\delta (\delta-s)\mathrm{Cov}\big(Y_0, Y_s\big)ds$$
and
$$\mathrm{Cov}(Y_i^{\delta}, Y_{j}^\delta) =  \delta \int_{-\infty}^{\infty}\left(1 - \frac{|s|}{\delta}\right)_+\mathrm{Cov}(Y_0,Y_{(j-i) \delta-s})\,ds.$$
It then suffices to compute explicitly $\mathrm{Cov}\big(Y_0, Y_s\big)$ for all three models. This follows from the representation \eqref{eq: cov generique} once the Bartlett spectrum in \eqref{eq: bartlett spectrum} has been specified. The Hawkes case is treated in Proposition \ref{prop: multiscale contrast}, while the necessary ingredients for the BL and NS models are given explicitly in \eqref{eq: specific BL} and \eqref{eq: specific NS}, respectively, in the proof of Proposition \ref{prop: fourier spectral} below.

\subsection{About the stationarity assumption} \label{sec: stationary assumption}
For both methods, it is more convenient to work with a stationary version of $(N_t)_{t \geq 0}$ continuated over negative time $t \leq 0$. This means that if 
$\mathcal N((s,t]) = N_t-N_s$ denotes the (uniquely defined) random measure over Borel sets of $[0,\infty)$, then it admits an extension over the whole real line such that for any Borel set $ A \subset \R$ and every $\tau \in \R$, we have $\mathcal N(A+\tau) = \mathcal N(A)$ in distribution.\\

As soon as the Hawkes process is stable in the sense of Definition \ref{def: stability}, such an extension always exists. One possible construction is the following: we start with a Poisson random measure $\Pi(ds,d\vartheta)$ on $\R\times [0,\infty)$ with intensity $ds\otimes d\vartheta$. We then set
$$\lambda_t = \mu+\int_{(s,\vartheta) \in (-\infty, t) \times [0,\infty)}{\bf 1}_{\{\vartheta \leq \lambda_s\}}\varphi(t-s)\Pi(ds,d\vartheta),\;\;t \in \R,$$ 
where we set $\varphi(t)=0$ for $t <0$. As soon as the model is stable, the process $(\lambda_t)_{t \in \mathbf{R}}$ is uniquely defined and  stationary (in the sense that $\lambda_{t+\tau} = \lambda_t$ in distribution for every $t,\tau \in \R$). The stationary random counting measure $\mathcal N$ is defined using the same realisation of the Poisson random measure $\Pi$ by
\[
\mathcal N(ds)
=
\int_{[0,\infty)}
{\bf 1}_{\{\vartheta\leq\lambda_{s}\}}
\Pi(ds,d\vartheta).
\]
The associated process $(N_t)_{t \in \R}$, anchored at $N_0=0$, is then defined by
\[
N_t
=
\begin{cases}
\mathcal N((0,t]), & t\geq0,\\[2mm]
-\mathcal N((t,0]), & t<0.
\end{cases}
\]
In particular, we have
$$\lambda_t =  \mu + \int_{(-\infty, t)} \varphi(t-s)dN_s,\;\;t \in \R.$$
Moreover, by \cite[Lemma 1.1.5]{bremaud2020}, there exists a random sequence $\{ T_m\}_{m \in \mathbb N}$ such that the point measure $\mathcal N(dt)$ has representation
$$N(dt) = \sum_{m \in \mathbb{N}} \delta_{T_m}(dt).$$
From a modelling point of view, assuming a stationary version means that the process we observe has some history in the past combined with a stability property. This is not a strong assumption in the context of rainfall modelling. In particular, in the representation \eqref{eq:intensity_rain}, we now have, for $t \in \R$,
\begin{equation} \label{eq: def models spectral stat}
Y_t  = \sum_{m \in \mathbb{N}} I_m \,{\bf 1}_{\{T_m \leq t < T_m + L_m\}},
\end{equation}
where $\{I_m, L_m\}_{m \in \mathbb N}$ are independent with common distribution and represent the rectangular pulse associated to the arrival time 
$T_m$ of a rain cell.

\subsection{Proof of Proposition \ref{prop: fourier spectral}} \label{app: proof prop fourier}

We take a stationary version of $(Y_t)_{t \in \R}$ using the representation \eqref{eq: def models spectral stat}. We plan to apply \cite[Theorem 9.4.1]{bremaud2020}. Introduce $h_t(u, (i,l)) = i {\bf 1}_{\{l> t-u\}}{\bf 1}_{\{t-u \geq 0\}}$ and $H_t(u) = \E[h_t(u, (I, L))]$. The assumptions of the theorem are verified as soon as $H_t \in L^1(\R) \cap L^2(\R)$, which is granted by the conditions $\E[L]<\infty$ and $\E[I^2]<\infty$. We obtain
\begin{align}
\mathrm{Cov}(Y_t, Y_0)  & =  \mathrm{Cov}\Big(\sum_{m \in \mathbb N} h_t(T_m, (I_m,L_m)), \sum_{m \in \mathbb N} h_0(T_m, (I_m,L_m))\Big) \nonumber \\
& = \int_{\R} \mathcal F H_t(2\pi \omega) \overline{\mathcal F H_0(2\pi \omega)}\rho(d\omega) \nonumber \\
&+ \Lambda \int_{\R} \mathrm{Cov}\big(\mathcal Fh_t(2\pi \omega, (I, L)), \overline{\mathcal Fh_0(2\pi \omega, (I, L))}\big) d\omega, \label{eq: cov generique}
 \end{align}
where 
\begin{equation} \label{eq: bartlett spectrum}
\rho(d\omega) = \frac{\Lambda}{|1-\mathcal F\varphi(2\pi \omega)|^2}d\omega
\end{equation}
is the Bartlett spectrum of a linear Hawkes process with parameters $(\mu, \varphi)$, see {\it e.g.} \cite[Theorem 12.3.1]{bremaud2020}.
From $\mathcal FH_t(2\pi\omega) = \exp(-\iota 2\pi\omega t)\mathcal FH_0(2\pi \omega)$, with $\iota^2=-1$, and likewise for $\mathcal Fh_t$, we further have that $\mathrm{Cov}(Y_t, Y_0)$ equals
\begin{align*}
&\Lambda \int_{\R}\frac{ |\mathcal F H_0(2\pi \omega)|^2 }{|1-\mathcal F\varphi(2\pi \omega)|^2}\exp(-\iota 2\pi \omega t)d\omega + \Lambda \int_{\R} \mathrm{Var}\big(\mathcal Fh_0(2\pi \omega, (I, L)\big)\exp(-\iota 2\pi \omega t) d\omega \\
& = \frac{\Lambda}{2\pi} \int_{\R}\frac{ |\mathcal F H_0(\omega)|^2 }{|1-\mathcal F\varphi(\omega)|^2}\exp(\iota\omega t)d\omega + \frac{\Lambda}{2\pi} \int_{\R} \mathrm{Var}\big(\mathcal Fh_0( \omega, (I, L))\big)\exp(\iota\omega t) d\omega.
\end{align*}
By Fourier inversion, we deduce
$$\mathcal F_Y(\omega) =  \Lambda \frac{ |\mathcal F H_0(\omega)|^2 }{|1-\mathcal F\varphi(\omega)|^2}+ \Lambda\mathrm{Var}\big(\mathcal Fh_0( \omega, (I, L)).$$
Now, since $\mathcal Fh_0(\omega, (I,L)) = I \iota \frac{1-\exp(\iota\omega L)}{\omega}$ we have
$$\mathcal FH_0(\omega) = \E[I] \iota \frac{1-\E[\exp(\iota\omega L)]}{\omega}\;\;\text{and}\;\;|\mathcal FH_0(\omega) |^2 =  \E[I]^2 \frac{|1-\E[\exp(\iota\omega L)]|^2}{\omega^2}.$$
Moreover
$$\mathrm{Var}\big(\mathcal Fh_0(\omega, (I,L))\big) = 2\E[I^2]\frac{1 - \mathrm{Re}(\E[\exp(\iota\omega L)])}{\omega^2}- |\mathcal FH_0(\omega)|^2,$$
and the result follows by elementary computations. The proof of Proposition \ref{prop: fourier spectral} is complete.\\

We now treat the case of the BS and NS models. We have the same formula, replacing formally $\frac{1}{|1 - \mathcal F \varphi(\omega)|^2}-1$ by $p(\omega)$. For the BL model, we have
\begin{equation} \label{eq: specific BL}
p(\omega) = 2\frac{\nu_{BL} \,\gamma_{BL}}{\gamma_{BL}^2+\omega^2},\;\;\Lambda = \mu_{BL}\Big(1+\frac{\nu_{BL}}{\gamma_{BL}}\Big).
\end{equation}
For the NS model, we have: 
\begin{equation} \label{eq: specific NS}
p(\omega) = \frac{\nu_{NS}\,\gamma_{NS}^2}{\gamma_{NS}^2+\omega^2},\;\;\Lambda  = \mu_{NS}\,\nu_{NS},
\end{equation}
in the case where $C$ follows a Poisson distribution with parameter $\nu_{NS}$ and $f_{NS}(t)dt$ is exponentially distributed with parameter $\gamma_{NS}$. These formulas are obtained with the same strategy as in the Hawkes case. The Bartlett spectrum \eqref{eq: bartlett spectrum} becomes
$$\rho(d\omega)  = \Lambda (1+p(\omega))d\omega,$$
applying the Fourier transform to the covariance formulas given right after (4.11) for the BL model and in (3.3) for the NS model in \cite{RodriguezIturbe1987}.

%

\subsection{Proof of Proposition \ref{prop: multiscale contrast}} \label{app: proof of multiscale contrast}
We have, for $\omega \in \R$,
\[
\mathcal{F}\varphi(\omega) = \sum_{i=1}^d \frac{\alpha_i}{\beta_i+\iota\omega}, 
\]
and hence
\begin{equation}
    \label{eq:fraction_bartlett}
\frac{1}{|1-\mathcal{F}\varphi(\omega)|^2} = \frac{\prod_{i=1}^d (\beta_i^2+\omega^2)}{D(\omega)},
\end{equation}
with 
   $$
D(\omega)  = \Big(\prod_{i=1}^d (\beta_i-\iota\omega) - \sum_{j=1}^d \alpha_j\prod_{\substack{i=1\\i\neq j}}^d (\beta_i-\iota\omega)\Big)\Big(\prod_{i=1}^d (\beta_i+\iota\omega) - \sum_{j=1}^d \alpha_j\prod_{\substack{i=1\\i\neq j}}^d(\beta_i+\iota\omega)\Big),
$$
and $D(\omega)$ is a positive polynomial function with degree $2d$ hence with $2d$ non-real complex roots. Since the $\beta_i$ are distinct and $\sum_{i = 1}^d \alpha_i/\beta_i < 1$, the roots of $D$ are different from $\pm \iota \beta_i$. Moreover, it follows from the factorisation of $D$ that every root $\omega$ satisfies
\[
1 - \sum_{j=1}^d \frac{\alpha_j}{\beta_j - \iota \omega} = 0 \text{ or } 1 - \sum_{j=1}^d \frac{\alpha_j}{\beta_j + \iota \omega} = 0
\]
therefore a root $\omega$ has a vanishing real part. The roots of $D(\omega)$ can then be written as $\pm \iota p_i$ with $p_i \in \mathbb{R}$ for $i=1,\ldots,d$, and the $p_i$ are the roots of 
\[
P(\omega) = \prod_{i=1}^d (\beta_i-\omega) - \sum_{j=1}^d \alpha_j\prod_{\substack{i=1\\i\neq j}}^d (\beta_i-\omega)
\]
since $D(\iota\omega) = P(\omega)P(-\omega)$. Since $\sum_{i=1}^d \frac{\alpha_i}{\beta_i} < 1$,
\[
P(\omega) = \prod_{i=1}^d (\beta_i-\omega)\Big(1 - \sum_{j=1}^d \frac{\alpha_j}{\beta_j-\omega}\Big) > 0 \text{ for } \omega \leq 0
\]
and the roots $p_i \in \mathbb{R}$ of $P$ are all positive. Without loss of generality, assuming $\beta_1 < \ldots < \beta_d$, the function 
\[
f(\omega) = 1 - \sum_{j=1}^d \frac{\alpha_j}{\beta_j-\omega}
\]
is strictly decreasing on each interval $(\beta_{i-1},\beta_{i})$, $i=1,\ldots,d$, where $\beta_0 = 0$. Moreover, $f(\beta_0) = 1 - \sum_{j=1}^d \frac{\alpha_j}{\beta_j} > 0$, $\lim_{\omega\uparrow\beta_i} f(\omega)=-\infty$, and $\lim_{\omega\downarrow\beta_i} f(\omega)=\infty$ for $i=1,\ldots,d$. Therefore, there is a unique root in each interval $(\beta_{i-1}, \beta_{i})$ for $i=1,\ldots,d$ and the roots are pairwise distinct.

\medskip
We can then rewrite \eqref{eq:fraction_bartlett} as
\[
\frac{1}{|1-\mathcal{F}\varphi(\omega)|^2}  = Q(\omega) + \sum_{i=1}^d \Big(\frac{E_i}{\omega - \iota p_i} + \frac{F_i}{\omega + \iota p_i}\Big).
\]
Since the left-hand side of \eqref{eq:fraction_bartlett} converges to $1$ as $|\omega| \to \infty$, the polynomial $Q$ is identically equal to $1$. We have
$E_i = \lim_{\omega \to \iota p_i} \tfrac{\prod_{j=1}^d (\beta_j^2+\omega^2)(\omega - \iota p_i)}{D(\omega)}$. Equivalently,
$E_i = -\iota A_i$ with 
\[
A_i = -\frac{\prod_{j=1}^d (\beta_j^2-p_i^2)}{P(-p_i)}\lim_{\omega \to p_i}\frac{\omega - p_i}{P(\omega)} = -\frac{\prod_{j=1}^d (\beta_j^2-p_i^2)}{P(-p_i)}\frac{1}{P'(p_i)}.
\]
Moreover $F_i = \lim_{\omega \to -\iota p_i} \frac{\prod_{j=1}^d (\beta_j^2+\omega^2)(\omega + \iota p_i)}{D(\omega)}$ is the conjugate of $E_i$ and therefore equals $\iota A_i$. We finally get
\begin{equation} \label{eq:decomp_fraction}
 \frac{1}{|1-\mathcal{F}\varphi(\omega)|^2} = 1 + \sum_{i=1}^d \frac{2A_i p_i}{\omega^2 + p_i^2}.
\end{equation}
By Proposition \ref{prop: fourier spectral}, together with
$
\E[\exp(\iota \omega L)] = \frac{\lambda_L}{\lambda_L - \iota \omega}
$
and \eqref{eq:decomp_fraction}, we obtain 
\[
\begin{split}
\mathcal{F}_Y(\omega) 
&= \frac{\Lambda}{\omega^2 + \lambda_L^2} \Big(2\mathbb{E}[I^2] 
+ \mathbb{E}[I]^2\sum_{i=1}^d \frac{2A_i p_i}{\omega^2 + p_i^2} \Big)\\
&= \Lambda \lambda_L^{-1}\Big(\mathbb{E}[I^2] 
- \mathbb{E}[I]^2 \sum_{i=1}^d \frac{A_i p_i}{\lambda_L^2 - p_i^2}\Big)
\frac{2\lambda_L}{\omega^2 + \lambda_L^2} 
+ \Lambda\sum_{i=1}^d \frac{\mathbb{E}[I]^2 A_i}{\lambda_L^2 - p_i^2}
\frac{2p_i}{\omega^2+p_i^2} \\
& = \Lambda \sum_{i=1}^d C_i \frac{2p_i}{p_i^2+\omega^2} +\Lambda C_{d+1}\frac{2\lambda_L}{\lambda_L^2+\omega^2}.
\end{split}
\]
Fourier inversion yields
\begin{equation} \label{eq: cov continue}
\mathrm{Cov}\big(Y_0, Y_t\big) = \Lambda  \sum_{i=1}^d C_i \exp(-p_i|t|) +\Lambda C_{d+1}\exp(-\lambda_L|t|),\;\;t \in \R.
\end{equation}
We are ready to establish the variance formula of Proposition \ref{prop: multiscale contrast}. By stationarity, we have
\begin{align}
\mathrm{Var}(Y_k^h)  = \mathrm{Var}(Y_0^h) &  = \int_0^h\int_0^h\mathrm{Cov}\big(Y_s, Y_t\big)ds\,dt = 2 \int_0^h\int_0^t\mathrm{Cov}\big(Y_s, Y_t\big)ds\,dt \nonumber\\
& =  2 \int_0^h\int_0^t\mathrm{Cov}\big(Y_0, Y_{t-s}\big)ds\,dt =  2 \int_0^h (h-s)\mathrm{Cov}\big(Y_0, Y_s\big)ds.   \label{eq: var en fourier}
\end{align}
From the elementary identity
$$2\int_0^h \exp(-u s) (h-s)ds = \frac{2h}{u}\Big(1- \frac{1-\exp(-uh)}{uh}\Big),\;\;u>0,$$
and \eqref{eq: cov continue}, we obtain the variance formula. For the covariance formula, we have likewise
\begin{align}
\mathrm{Cov}(Y_k^{h}, Y_{k'}^h) &= \int_{(k'-1)h}^{k' h} \int_{(k-1)h}^{k h} \mathrm{Cov}(Y_s,Y_t)\,ds\,dt = \int_{(k'-1)h}^{k' h} \int_{t-kh}^{t-(k-1)h} \mathrm{Cov}(Y_0,Y_{s})\,ds\, dt \nonumber \\
&=  \int_{(k'-k-1)h}^{(k'-k+1)h}\mathrm{Cov}(Y_0,Y_{s}) \int_{(s+(k-1)h) \vee (k'-1)h}^{(s+kh)\wedge k'h}dt \, ds \nonumber \\
&= h \int_{(k'-k-1)h}^{(k'-k+1)h}\left(1 - \frac{|(k'-k)h-s|}{h}\right)\mathrm{Cov}(Y_0,Y_{s}) \, ds \nonumber \\
&= h \int_{-\infty}^{\infty}\left(1 - \frac{|s|}{h}\right)_+\mathrm{Cov}(Y_0,Y_{(k'-k) h-s})\,ds. \label{eq: cov en fourier}
\end{align}
The covariance formula then follows from \eqref{eq: cov en fourier} and the identity
$$
\int_{-\infty}^{\infty}\Big(1 - \frac{|s|}{h}\Big)_+ \exp(-u|\tau h- s|) ds = \frac{\exp(-u(\tau-1)h)}{u^2h}\big(1-\exp(-uh) \big)^2,\;\;u>0,\;\;\tau \geq 1.
$$
The proof of Proposition \ref{prop: multiscale contrast} is complete.

\subsection{Proof of Lemma \ref{lemma: error bound}} \label{app: proof of error bound}

We have
\begin{equation} \label{eq: first bound xi}
|\xi_{t} |  =  \sum_{i \geq 1} I_i (\tau_i+L_i-t){\bf 1}_{\{\tau_i \leq t \leq \tau_i+L_i\}} 
 \leq  \sum_{i \geq 1} I_i L_i{\bf 1}_{\{\tau_i \leq t \leq \tau_i+L_i\}}.
\end{equation}
For computational purposes, it will be convenient to use a Poisson measure representation; to that end, let $\Pi(ds\,d\vartheta, di, d\ell)$ be a Poisson random measure on $[0,\infty)^4$ with intensity $ds\,d\vartheta \otimes f_I(di) \otimes f_L(d\ell)$, where $f_I$ and $f_L$ denote the distributions of $I$ and $L$, respectively. If the intensity $(\lambda_t)_{t \geq 0}$ of $(N_t)_{t \geq 0}$ is realized with the same Poisson measure, namely
$$\lambda_t = \mu+\int_{[0,\infty)^2}{\bf 1}_{\{\vartheta \leq \lambda_s\}}\varphi(t-s)\widetilde{\Pi}(ds\,d\vartheta),$$
where $\widetilde{\Pi}(ds,d\vartheta) = \int_{[0,\infty)^2}\Pi(ds\,d\vartheta, di, d\ell)$is a standard Poisson random measure on $[0,\infty)^2$ with intensity $ds\,d\vartheta$,
we then have
$$\sum_{i \geq 1} I_i L_i{\bf 1}_{\{\tau_i \leq t \leq \tau_i+L_i\}} = \int_0^t \int_{[0, \infty)^3} i\ell {\bf 1}_{\{\vartheta \leq \lambda_s^T\}}{\bf 1}_{\{s \leq t \leq s+\ell\}}\Pi(ds\,d\vartheta, di, d\ell).$$ 
By \eqref{eq: first bound xi}, it follows that
\begin{align} 
\E\big[|\xi_{t} | \big] & \leq \int_0^{t} \int_{[0,\infty)}\E[I]\E[\lambda_s]\ell {\bf 1}_{\{s \leq t \leq s+\ell\}} f_L(d\ell) ds  \nonumber\\
& = \E[I] \int_0^{t}\E[\lambda_s]\int_{t-s}^\infty\ell f_L(d\ell) ds. \label{eq: second bound xi}
\end{align}
Introducing $\Psi(s)  = \sum_{n \geq 1} \varphi^{\star n}(s)$, where $\star n$ denotes $n$-fold convolution,
we have
\begin{equation} \label{eq: renewal}
\E[\lambda_s] = \mu\Big(1+\int_0^s \Psi(u)du\Big) \leq \mu\Big(1+\int_0^\infty \Psi(u)du\Big),
\end{equation}
as the solution of the renewal equation
$\E[\lambda_t]  = \mu +\int_0^t \varphi(t-s)\E[\lambda_s]ds$
that follows from the very definition $\lambda_t = \mu + \int_{[0,t)} \varphi(t-s)dN_s$ and the fact that $N_t-\int_0^t \lambda_sds$ is a (local) martingale. 
Injecting \eqref{eq: renewal} into \eqref{eq: second bound xi} yields
\begin{align*}
\E\big[|\xi_{t} | \big] & \leq \E[I]\mu\Big(1+\int_0^\infty \Psi(u)du\Big) \int_0^{t}\int_{t-s}^\infty\ell f_L(d\ell) ds \\
& \leq  \E[I]\mu\Big(1+\int_0^\infty \Psi(u)du\Big) \int_0^{\infty}\int_{s}^\infty\ell f_L(d\ell) ds \\
& = \E[I]\mu\Big(1+\int_0^\infty \Psi(u)du\Big)\E[L^2]. 
\end{align*}
Finally, 
$$\mu\Big(1+\int_0^\infty \Psi(u)du\Big) = \mu\big(1+\sum_{n \geq 1}\|\varphi\|_{1}^n\big) = \frac{\mu}{1-\|\varphi\|_{1}} \lesssim T^{2\alpha-1}$$
as follows from the criticality condition \eqref{eq: kernel expansion}, and this establishes the first bound. For the second bound, for $t\geq 0$, introduce
$$\zeta_{t} = \sum_{i = 1}^{N_{t}}(I_iL_i-\kappa).$$
By independence of $(N_{t})_{t \geq 0}$ and $(I_i,L_i)_{i \geq 1}$, $\zeta_{t}$ is centered. It then suffices to show that $\E[\zeta_{tT}^2] \lesssim T^{2\alpha}$ uniformly in $t \in [0,1]$. By independence again, we have
\begin{align*}
\E\big[\zeta_{tT}^2\big]   \leq \E[I^2]\E[L^2] \E[N_{tT}]  \leq \E[I^2]\E[L^2] \int_0^{tT} \E[\lambda_s]ds 
 \lesssim \frac{tT\mu}{1-\|\varphi\|_{1}} \lesssim T^{2\alpha}
\end{align*}
using the same estimates derived from \eqref{eq: renewal} above. The proof of Lemma \ref{lemma: error bound} is complete.

\subsection{Proof of Proposition \ref{prop: small meets large}} \label{sec: proof of small meets large}
To avoid trivialities, we assume $\Delta T \geq \delta$ or, in terms of sample size, $n \geq N$. First, we rescale large time data over $[0,T]$ in order to compare times series over the standardised time interval $[0,T]$ using \eqref{eq: continuous embedding}. For $k=1,\ldots, N$, a large time data $Z_{k}^\Delta$ measures the aggregation of rainfall over the (standardised) time interval 
\begin{align*}
&\big[(k-1)\Delta T, k\Delta T\big] = \bigcup_{\ell\;\text{such that}\;(\ell-1)\delta = (k-1)\Delta T}^{\ell\;\text{such that}\;\ell \delta= k\Delta T} [(\ell-1)\delta, \ell\delta].
\end{align*}
Ignoring boundary issues and assuming that $\Delta T/\delta$ is an integer, we thus have the correspondence
\begin{align*}
Z_k^\Delta  = \sum_{\ell = (k-1)\Delta T/\delta+1}^{k\Delta T/\delta} Y_\ell^\delta  = \int_{(k-1)\Delta T}^{k\Delta T} Y_s\,ds,
\end{align*}
using the aggregation representation $Y_\ell^\delta = \int_{(\ell-1)\delta}^{\ell \delta} Y_s\,ds$ in \eqref{eq: data micro}, and where $(Y_t)_{t \in [0,T]}$ is the latent intensity process \eqref{eq:intensity_rain}.  Now, 
let $(Y_t)_{t \in [0,T]}$ be driven by a linear critical Hawkes process as in Section \ref{sec: connecting hawkes and fractional} above. The following approximation becomes valid:
\begin{align*}
T^{-2\alpha}\int_{(k-1)\Delta T}^{k\Delta T} Y_sds  
& =  \kappa \int_{(k-1)\Delta}^{k\Delta} X_s\,ds+o(1)
\end{align*}
by Lemma \ref{lemma: error bound} and \eqref{eq: serious scale}.
By the continuity of the sample paths of $(X_t)_{t \in [0,1]}$, in the limit $\Delta \rightarrow 0$ and $T \rightarrow \infty$  we obtain, thanks to the interpolation \eqref{eq: continuous embedding}
\begin{align*}
T^{-2\alpha} Z_{k\Delta}  = T^{-2\alpha} Z_k^\Delta & =  \kappa \int_{(k-1)\Delta}^{k\Delta} X_s\,ds+o(1) \\
 & = \kappa \Delta X_{k\Delta}  + o(1),
\end{align*}
which yields 
$$T^{-2\alpha} Z_{t} = \kappa \Delta X_{t}  + o(1)$$
for $t \in [0,1]$ by the continuity of the path of $(X_t)_{t \in [0,1]}$ and  $(Z_t)_{t \in [0,1]}$.
Having $(X_t)_{t \in [0,1]}$ to be a fractional process with Hurst index $H$, the same result holds for $(Z_t)_{t \in [0,1]}$. The critical exponent $\alpha$ of the small time model of the intensity process $(Y_t)_{t \in [0,T]}$ has a large time trace via the Hurst index $H =\alpha-1/2$ in the fractional limit. The proof of Proposition \ref{prop: small meets large} is complete.

\subsection{Proof of Proposition \ref{prop: frac property}} \label{app: proof of frac property}
We first prove the upper bound in \eqref{eq: scaling property gen}. We plan to apply \cite{AbiJaberLarssonPulido2019}. By assumption, for $h >0$, the bound on $k_H$ implies, 
\begin{equation} \label{eq: small h}
(k_H(t+h)-k_H(t))^2 \leq C_H^2((t+h)^{H-1/2}+t^{H-1/2})^2 \lesssim h^{2H-1}\;\;\text{for}\;\;t \in [0,h].
\end{equation}
For $t \geq h$, writing 
$$k_H(t+h)-k_H(t) = k_H(t)\Big(\frac{k_H(t+h)}{k_H(t)}-1\Big),$$
from\footnote{The symbol $\sim$ means inequality in both ways, up to constants that may depend on $c_H, C_H$ and $H$ only.} 
$$\frac{k_H(t+h)}{k_H(t)} \sim \Big(1+\frac{h}{t}\Big)^{H-1/2},$$
and
the elementary bound $|(1+x)^{H-1/2}-1| \sim x$ valid for $H \in (0,1) \setminus \{\tfrac{1}{2}\}$ and $x \in (0,1)$, we derive
$$\Big|\frac{k_H(t+h)}{k_H(t)}-1\Big| \sim \frac{h}{t}\;\;\text{for}\;\;t \geq h..$$
It follows that 
\begin{equation} \label{eq: big h}
(k_H(t+h)-k_H(t))^2 \lesssim t^{2H-3}h^2\;\;\text{for}\;\;t \geq h.
\end{equation}
Putting together \eqref{eq: small h} and \eqref{eq: big h}, we obtain for $T>0$ and small enough $h$:
\begin{equation} \label{eq: eduardo 1}
\int_0^T(k_H(t+h)-k_H(t))^2dt \lesssim h^{2H-1}h+h^2\int_h^T t^{2H-3}dt \lesssim h^{2H}.
\end{equation}
Moreover, we readily have $\int_0^h k_H(t)^2dt \lesssim h^{2H}$. The kernel $k_H$ satisfies the crucial condition (2.5) of \cite{{AbiJaberLarssonPulido2019}} and the upper bound of  Proposition \ref{prop: frac property} readily follows from the proof of \cite[Lemma 2.4]{AbiJaberLarssonPulido2019}, see in particular the bound right after the estimate (2.10) in the proof of this lemma. Moreover, for every $T>0$, the lemma entails $\sup_{t \in [0,T]}\E[|h(Z_t)|^q] < \infty$ for some $q>2$, a bound that we will need later on.\\ 

We next prove the lower bound in \eqref{eq: scaling property gen} under the restriction $q \geq 2$.  For $t\geq 0$, introduce 
$$\xi_t = \int_{0}^t (k_H(t-u)-k_H(s-u){\bf 1}_{\{u \leq s\}})h(Z_u)dB_u.$$
From 
$$\E\big[|Z_t-Z_s|^q\big] \geq 2^{1-q}\E\big[|\xi_t-\xi_s|^q\big]-|g(t)-g(s)|^q$$
and the fact that $|g(t)-g(s)|^q \lesssim |t-s|^q$ by Lipschitz continuity, it suffices to prove the bound for $(\xi_t)_{t \geq 0}$.  
The random process
$$t'\mapsto \int_{0}^{t'} (k_H(t-u)-k_H(s-u){\bf 1}_{\{u \leq s\}})h(Z_u)dB_u,\;\;\text{for fixed}\;\;0 \leq s \leq t,$$
is a local martingale. The Burkholder-Davis-Gundy inequality at $t'=t$ and Jensen's inequality yields
\begin{align*}
\E\big[|\xi_t-\xi_s|^q\big]  & \geq c_q \E\Big[\big(\int_0^t (k_H(t-u)-k_H(s-u){\bf 1}_{\{u \leq s\}})^2h(Z_u)^2du\big)^{q/2}\Big] \\
& \geq c_q \Big(\int_0^t (k_H(t-u)-k_H(s-u){\bf 1}_{\{u \leq s\}})^2\E[h(Z_u)^2]du\Big)^{q/2}.
\end{align*}
From the upper bound, by Kolmogorov's continuity criterion,  we have that $t \mapsto h(Z_t)^2$ has a continuous modification so that $t \mapsto \E[h(Z_t)^2]$ is continuous likewise by the uniform integrability of the family $h(Z_t)_{t \in [0,1]}$ granted by $\sup_{t \in [0,1]}\E[|Z_t|^q]< \infty$ for every $q \geq 2$,  that follows from polynomial growth of $h$ and the upper bound. Since $\E[h(Z_t)^2] \neq 0$ for every $t$, we have that $c_h = \inf_{u \in [0,1]}\E[h(Z_u)^2] >0$. It follows that 
\begin{align*}
\E\big[|\xi_t-\xi_s|^q\big]  & \geq c_qc_h\Big(\int_0^t (k_H(t-u)-k_H(s-u){\bf 1}_{\{u \leq s\}})^2du\Big)^{q/2} \\
& \geq c_qc_h\Big(\int_{s}^t k_H(t-u)^2du\Big)^{q/2} \\
& \geq c_qc_hc_H^q\Big(\int_{s}^t (t-u)^{2H-1}du\Big)^{q/2} \\
& = \frac{c_qc_hc_H^q}{(2H)^{q/2}}(t-s)^{Hq},
\end{align*}
and we obtain the lower bound with $k_q = \frac{c_qc_hc_H^q}{(2H)^{q/2}}$. The proof of Proposition \ref{prop: frac property} is complete.

\subsection{Detailed results for M\'et\'eo France data at small time scales}
\label{app:mf_micro_results}
In this section, we report the different results obtained for the four M\'et\'eo France weather stations at small time scales, in addition to results of the Section~\ref{sec: stat criticality}.

\begin{table}[H]
\centering
\begin{tabular}{lllllll}
\toprule
Month & $\|\varphi\|_1$ & $\sigma_{\|\varphi\|_1}$ & $\alpha$ & $\sigma_{\alpha}$ & $n_{>0}$ & $\frac{\mu T}{1-\|\varphi\|_1}$ \\
\midrule
Jan & 0.961 & 0.009 & 0.591 & 0.066 & \num{4.74e+03} & \num{7.19e+04} \\
Feb & 0.952 & 0.009 & 0.662 & 0.060 & \num{4.28e+03} & \num{4.88e+04} \\
Mar & 0.917 & 0.020 & 0.451 & 0.106 & \num{4.15e+03} & \num{6.22e+03} \\
\midrule
Apr & 0.851 & 0.037 & 0.417 & 0.120 & \num{2.56e+03} & \num{2.75e+03} \\
May & 0.936 & 0.012 & 0.507 & 0.090 & \num{3.64e+03} & \num{1.96e+04} \\
Jun & 0.957 & 0.022 & 0.456 & 0.111 & \num{3.24e+03} & \num{3.71e+04} \\
\midrule
Jul & 0.913 & 0.025 & 0.392 & 0.099 & \num{2.98e+03} & \num{7.58e+03} \\
Aug & 0.716 & 0.041 & 0.815 & 1.621 & \num{3.51e+03} & \num{1.53e+03} \\
Sep & 0.914 & 0.030 & 0.372 & 0.124 & \num{2.92e+03} & \num{6.82e+03} \\
\midrule
Oct & 0.859 & 0.024 & 0.853 & 1.337 & \num{4.96e+03} & \num{6.45e+03} \\
Nov & 0.977 & 0.006 & 0.473 & 0.072 & \num{5.70e+03} & \num{1.61e+05} \\
Dec & 0.976 & 0.006 & 0.494 & 0.059 & \num{5.16e+03} & \num{1.58e+05} \\
\bottomrule
\end{tabular}
\caption{{\it \small Parameter estimates for $\alpha$ and $\|\varphi\|_{1}$ for the Lille station under the Hawkes model with the power-law approximation \eqref{eq: power-law approximation} based on the second-order contrast method, with standard deviation 
$\sigma_{\|\varphi\|_{1}}$ or $\sigma_\alpha$ based 
on 100 repeated simulations with estimated parameters. The number $n_{>0}$ indicates the number of non-zero data. The last column displays a proxy of the statistical information (in number of events) $\frac{\mu T}{1-\|\varphi\|_1}$, where $\mu$ and $\|\varphi\|_{1}$ are replaced by our estimators.}  \label{table:lille_ar1}
}
\end{table}

\begin{table}[H]
\centering
\begin{tabular}{lllllll}
\toprule
Month & $\|\varphi\|_1$ & $\sigma_{\|\varphi\|_1}$ & $\alpha$ & $\sigma_{\alpha}$ & $n_{>0}$ & $\frac{\mu T}{1-\|\varphi\|_1}$ \\
\midrule
Jan & 0.925 & 0.010 & 1.031 & 0.211 & \num{4.59e+03} & \num{2.88e+04} \\
Feb & 0.925 & 0.010 & 0.950 & 0.083 & \num{4.14e+03} & \num{2.57e+04} \\
Mar & 0.926 & 0.010 & 0.990 & 0.130 & \num{4.15e+03} & \num{3.08e+04} \\
\midrule
Apr & 0.878 & 0.022 & 0.903 & 0.166 & \num{2.56e+03} & \num{8.72e+03} \\
May & 0.876 & 0.030 & 0.692 & 0.132 & \num{3.64e+03} & \num{7.18e+03} \\
Jun & 0.902 & 0.018 & 1.311 & 1.494 & \num{3.24e+03} & \num{1.73e+04} \\
\midrule
Jul & 0.865 & 0.031 & 0.627 & 0.113 & \num{2.98e+03} & \num{5.75e+03} \\
Aug & 0.983 & 0.013 & 0.852 & 0.109 & \num{3.51e+03} & \num{6.36e+05} \\
Sep & 0.942 & 0.030 & 0.550 & 0.181 & \num{2.82e+03} & \num{3.05e+04} \\
\midrule
Oct & 0.921 & 0.020 & 0.600 & 0.074 & \num{4.96e+03} & \num{1.76e+04} \\
Nov & 0.931 & 0.009 & 0.929 & 0.104 & \num{5.04e+03} & \num{3.77e+04} \\
Dec & 0.922 & 0.012 & 0.760 & 0.069 & \num{4.12e+03} & \num{2.68e+04} \\
\bottomrule
\end{tabular}\caption{{\it {\small Same experiment as in Table \ref{table:lille_ar1}, except that the spectral method is used for parameter estimation instead of the second-order contrast method.}}}
\end{table}

\begin{table}[H]
\centering
\begin{tabular}{lllllll}
\toprule
Month & $\|\varphi\|_1$ & $\sigma_{\|\varphi\|_1}$ & $\alpha$ & $\sigma_{\alpha}$ & $n_{>0}$ & $\frac{\mu T}{1-\|\varphi\|_1}$ \\
\midrule
Jan & 0.975 & 0.009 & 0.688 & 0.112 & \num{2.66e+03} & \num{4.02e+04} \\
Feb & 0.969 & 0.011 & 0.709 & 0.124 & \num{2.24e+03} & \num{4.44e+04} \\
Mar & 0.925 & 0.027 & 0.886 & 0.575 & \num{2.46e+03} & \num{6.76e+03} \\
\midrule
Apr & 0.977 & 0.011 & 0.576 & 0.130 & \num{2.30e+03} & \num{7.43e+04} \\
May & 0.960 & 0.032 & 0.466 & 0.130 & \num{1.91e+03} & \num{1.62e+04} \\
Jun & 0.880 & 0.252 & 0.159 & 1.067 & \num{8.63e+02} & \num{3.16e+02} \\
\midrule
\textcolor{red}{Jul} & 0.881 & 0.251 & 0.854 & 2.081 & \num{3.33e+02} & \num{8.94e+02} \\
Aug & 0.591 & 0.199 & 1.031 & 2.507 & \num{5.96e+02} & \num{1.48e+02} \\
Sep & 0.773 & 0.087 & 0.395 & 0.565 & \num{1.65e+03} & \num{5.14e+02} \\
\midrule
Oct & 0.970 & 0.017 & 0.524 & 0.139 & \num{2.61e+03} & \num{3.49e+04} \\
Nov & 0.865 & 0.027 & 0.581 & 0.314 & \num{3.55e+03} & \num{2.52e+03} \\
Dec & 0.932 & 0.025 & 1.086 & 1.913 & \num{2.76e+03} & \num{7.88e+03} \\
\bottomrule
\end{tabular}
\caption{{\it {\small Parameter estimates for $\alpha$ and $\|\varphi\|_{1}$ for the Marseille station under the Hawkes model with the power-law approximation \eqref{eq: power-law approximation} based on the second-order contrast method, with standard deviation 
$\sigma_{\|\varphi\|_{1}}$ or $\sigma_\alpha$ based 
on 100 repeated simulations with estimated parameters. The number $n_{>0}$ indicates the number of non-zero data. The last column displays a proxy of the statistical information (in number of events) $\frac{\mu T}{1-\|\varphi\|_1}$, where $\mu$ and $\|\varphi\|_{1}$ are replaced by our estimators. Months shown in red indicate cases where model does not achieve the lowest score~\eqref{eq:contrast} among the Hawkes model with an exponential kernel, the BL model, and the NS model.}}  \label{table:marseille_ar1}
}
\end{table}

\begin{table}[H]
\centering
\begin{tabular}{lllllll}
\toprule
Month & $\|\varphi\|_1$ & $\sigma_{\|\varphi\|_1}$ & $\alpha$ & $\sigma_{\alpha}$ & $n_{>0}$ & $\frac{\mu T}{1-\|\varphi\|_1}$ \\
\midrule
Jan & 0.957 & 0.011 & 1.021 & 0.140 & \num{2.66e+03} & \num{2.58e+04} \\
Feb & 0.962 & 0.011 & 0.800 & 0.189 & \num{2.24e+03} & \num{3.87e+04} \\
Mar & 0.965 & 0.010 & 0.583 & 0.098 & \num{2.46e+03} & \num{2.98e+04} \\
\midrule
Apr & 0.954 & 0.014 & 0.733 & 0.157 & \num{2.30e+03} & \num{2.51e+04} \\
May & 0.913 & 0.032 & 0.614 & 0.136 & \num{1.91e+03} & \num{3.85e+03} \\
\textcolor{red}{Jun} & 0.909 & 0.036 & 0.791 & 0.295 & \num{8.63e+02} & \num{3.41e+03} \\
\midrule
\textcolor{red}{Jul} & 0.994 & 0.004 & 10.000 & 1.939 & \num{3.33e+02} & \num{5.09e+05} \\
\textcolor{red}{Aug} & 0.917 & 0.069 & 1.348 & 2.733 & \num{5.96e+02} & \num{2.66e+03} \\
\textcolor{red}{Sep} & 0.940 & 0.011 & 10.000 & 2.598 & \num{1.65e+03} & \num{1.43e+04} \\
\midrule
\textcolor{red}{Oct} & 0.831 & 0.048 & 1.157 & 3.460 & \num{2.61e+03} & \num{2.02e+03} \\
Nov & 0.906 & 0.029 & 0.964 & 0.209 & \num{3.55e+03} & \num{8.74e+03} \\
Dec & 0.955 & 0.014 & 0.814 & 0.142 & \num{2.76e+03} & \num{1.67e+04} \\
\bottomrule
\end{tabular}\caption{{\it {\small Same experiment as in Table~\ref{table:marseille_ar1}, except that the spectral method is used for parameter estimation instead of the second-order contrast method. Unlike Table~\ref{table:marseille_ar1}, months shown in red indicate cases where the model does not achieve the lowest AIC score~\eqref{eq:aic} among the Hawkes model with an exponential kernel, the Bartlett-Lewis (BL) model, and the Neyman-Scott (NS) model, rather than cases where it does not achieve the lowest score~\eqref{eq:contrast}.}}}
\end{table}

\begin{table}[H]
\centering
\begin{tabular}{lllllll}
\toprule
Month & $\|\varphi\|_1$ & $\sigma_{\|\varphi\|_1}$ & $\alpha$ & $\sigma_{\alpha}$ & $n_{>0}$ & $\frac{\mu T}{1-\|\varphi\|_1}$ \\
\midrule
Jan & 0.944 & 0.008 & 0.676 & 0.070 & \num{3.52e+03} & \num{2.44e+04} \\
Feb & 0.917 & 0.015 & 0.727 & 0.086 & \num{2.80e+03} & \num{1.39e+04} \\
Mar & 0.946 & 0.010 & 0.561 & 0.057 & \num{3.09e+03} & \num{2.50e+04} \\
\midrule
Apr & 0.949 & 0.015 & 0.503 & 0.091 & \num{3.02e+03} & \num{2.06e+04} \\
May & 0.781 & 0.043 & 0.612 & 1.094 & \num{4.60e+03} & \num{2.11e+03} \\
Jun & 0.892 & 0.038 & 0.381 & 0.134 & \num{3.42e+03} & \num{6.81e+03} \\
\midrule
Jul & 0.819 & 0.025 & 0.518 & 0.162 & \num{3.08e+03} & \num{2.30e+03} \\
Aug & 0.753 & 0.038 & 0.881 & 1.172 & \num{3.36e+03} & \num{1.92e+03} \\
Sep & 0.927 & 0.020 & 0.421 & 0.098 & \num{2.87e+03} & \num{9.54e+03} \\
\midrule
Oct & 0.949 & 0.014 & 0.580 & 0.074 & \num{3.88e+03} & \num{1.67e+04} \\
Nov & 0.934 & 0.013 & 0.821 & 0.200 & \num{3.78e+03} & \num{1.98e+04} \\
Dec & 0.959 & 0.009 & 0.533 & 0.050 & \num{3.53e+03} & \num{3.94e+04} \\
\bottomrule
\end{tabular}
\caption{{\it {\small Parameter estimates for $\alpha$ and $\|\varphi\|_{1}$ for the Strasbourg station under the Hawkes model with the power-law approximation \eqref{eq: power-law approximation} based on the second-order contrast method, with standard deviation 
$\sigma_{\|\varphi\|_{1}}$ or $\sigma_\alpha$ based 
on 100 repeated simulations with estimated parameters. The number $n_{>0}$ indicates the number of non-zero data. The last column displays a proxy of the statistical information (in number of events) $\frac{\mu T}{1-\|\varphi\|_1}$, where $\mu$ and $\|\varphi\|_{1}$ are replaced by our estimators.}}  \label{table:strasbourg_ar1}
}
\end{table}

\begin{table}[H]
\centering
\begin{tabular}{lllllll}
\toprule
Month & $\|\varphi\|_1$ & $\sigma_{\|\varphi\|_1}$ & $\alpha$ & $\sigma_{\alpha}$ & $n_{>0}$ & $\frac{\mu T}{1-\|\varphi\|_1}$ \\
\midrule
Jan & 0.926 & 0.013 & 0.896 & 0.103 & \num{3.52e+03} & \num{2.18e+04} \\
Feb & 0.913 & 0.013 & 0.838 & 0.088 & \num{2.79e+03} & \num{1.70e+04} \\
Mar & 0.902 & 0.012 & 0.856 & 0.126 & \num{3.09e+03} & \num{1.73e+04} \\
\midrule
Apr & 0.917 & 0.022 & 0.683 & 0.104 & \num{2.94e+03} & \num{1.07e+04} \\
May & 0.824 & 0.035 & 0.773 & 0.155 & \num{4.60e+03} & \num{5.18e+03} \\
Jun & 0.972 & 0.019 & 0.697 & 0.111 & \num{3.42e+03} & \num{1.94e+05} \\
\midrule
Jul & 0.888 & 0.047 & 0.330 & 0.097 & \num{3.03e+03} & \num{3.38e+03} \\
\textcolor{red}{Aug} & 0.859 & 0.032 & 0.524 & 0.117 & \num{3.36e+03} & \num{3.77e+03} \\
\textcolor{red}{Sep} & 0.877 & 0.042 & 0.516 & 0.126 & \num{2.87e+03} & \num{4.10e+03} \\
\midrule
Oct & 0.940 & 0.016 & 0.684 & 0.085 & \num{3.88e+03} & \num{1.87e+04} \\
Nov & 0.946 & 0.011 & 0.918 & 0.088 & \num{3.78e+03} & \num{2.74e+04} \\
Dec & 0.907 & 0.016 & 0.937 & 0.172 & \num{2.78e+03} & \num{1.50e+04} \\
\bottomrule
\end{tabular}\caption{{\it {\small Same experiment as in Table \ref{table:strasbourg_ar1}, except that the spectral method is used for parameter estimation instead of the second-order contrast method. Months shown in red indicate cases where the power-law Hawkes model does not achieve the lowest AIC score~\eqref{eq:aic} among the Hawkes model with an exponential kernel, the BL model, and the NS model.}}}
\end{table}

\begin{table}[H]
\centering
\begin{tabular}{lllllll}
\toprule
Month & $\|\varphi\|_1$ & $\sigma_{\|\varphi\|_1}$ & $\alpha$ & $\sigma_{\alpha}$ & $n_{>0}$ & $\frac{\mu T}{1-\|\varphi\|_1}$ \\
\midrule
Jan & 0.950 & 0.010 & 0.823 & 0.106 & \num{4.22e+03} & \num{2.51e+04} \\
Feb & 0.977 & 0.006 & 0.509 & 0.059 & \num{2.88e+03} & \num{9.70e+04} \\
Mar & 0.902 & 0.015 & 0.886 & 0.123 & \num{3.38e+03} & \num{9.96e+03} \\
\midrule
Apr & 0.961 & 0.012 & 0.410 & 0.102 & \num{3.56e+03} & \num{3.49e+04} \\
May & 0.794 & 0.034 & 0.489 & 0.258 & \num{4.17e+03} & \num{2.26e+03} \\
Jun & 0.816 & 0.124 & 0.288 & 0.228 & \num{2.55e+03} & \num{1.04e+03} \\
\midrule
Jul & 0.907 & 0.038 & 0.316 & 0.136 & \num{1.71e+03} & \num{3.22e+03} \\
Aug & 0.830 & 0.023 & 0.753 & 0.153 & \num{1.65e+03} & \num{1.69e+03} \\
Sep & 0.866 & 0.032 & 0.611 & 0.488 & \num{1.85e+03} & \num{2.28e+03} \\
\midrule
Oct & 0.794 & 0.069 & 0.689 & 0.433 & \num{2.66e+03} & \num{1.00e+03} \\
Nov & 0.973 & 0.006 & 0.524 & 0.062 & \num{3.93e+03} & \num{8.26e+04} \\
Dec & 0.965 & 0.010 & 0.437 & 0.069 & \num{3.38e+03} & \num{2.99e+04} \\
\bottomrule
\end{tabular}
\caption{{\it {\small Parameter estimates for $\alpha$ and $\|\varphi\|_{1}$ for the Toulouse station under the Hawkes model with the power-law approximation \eqref{eq: power-law approximation} based on the second-order contrast method, with standard deviation 
$\sigma_{\|\varphi\|_{1}}$ or $\sigma_\alpha$ based 
on 100 repeated simulations with estimated parameters. The number $n_{>0}$ indicates the number of non-zero data. The last column displays a proxy of the statistical information (in number of events) $\frac{\mu T}{1-\|\varphi\|_1}$, where $\mu$ and $\|\varphi\|_{1}$ are replaced by our estimators.}}  \label{table:toulouse_ar1}
}
\end{table}

\begin{table}[H]
\centering
\begin{tabular}{lllllll}
\toprule
Month & $\|\varphi\|_1$ & $\sigma_{\|\varphi\|_1}$ & $\alpha$ & $\sigma_{\alpha}$ & $n_{>0}$ & $\frac{\mu T}{1-\|\varphi\|_1}$ \\
\midrule
Jan & 0.953 & 0.009 & 0.850 & 0.071 & \num{4.22e+03} & \num{3.34e+04} \\
Feb & 0.936 & 0.017 & 0.707 & 0.098 & \num{2.65e+03} & \num{1.62e+04} \\
Mar & 0.976 & 0.011 & 0.485 & 0.056 & \num{3.38e+03} & \num{8.26e+04} \\
\midrule
Apr & 0.940 & 0.015 & 0.885 & 0.221 & \num{3.56e+03} & \num{4.85e+04} \\
May & 0.921 & 0.019 & 0.625 & 0.082 & \num{4.17e+03} & \num{1.51e+04} \\
\textcolor{red}{Jun} & 0.810 & 0.055 & 1.153 & 1.752 & \num{2.55e+03} & \num{3.02e+03} \\
\midrule
Jul & 0.973 & 0.026 & 0.683 & 0.152 & \num{1.71e+03} & \num{1.04e+05} \\
\textcolor{red}{Aug} & 0.743 & 0.048 & 0.822 & 3.925 & \num{1.65e+03} & \num{9.39e+02} \\
Sep & 0.872 & 0.038 & 0.631 & 0.187 & \num{1.85e+03} & \num{3.05e+03} \\
\midrule
\textcolor{red}{Oct} & 0.941 & 0.010 & 10.000 & 3.000 & \num{2.66e+03} & \num{2.75e+04} \\
Nov & 0.915 & 0.014 & 1.046 & 0.215 & \num{3.62e+03} & \num{1.81e+04} \\
Dec & 0.923 & 0.011 & 0.840 & 0.081 & \num{3.38e+03} & \num{2.33e+04} \\
\bottomrule
\end{tabular}\caption{{\it {\small Same experiment as in Table \ref{table:toulouse_ar1}, except that the spectral method is used for parameter estimation instead of the second-order contrast method. Months shown in red indicate cases where the power-law Hawkes model does not achieve the lowest AIC score~\eqref{eq:aic} among the Hawkes model with an exponential kernel, the BL model, and the NS model.
}}}
\end{table}

\end{document}